\documentclass[11pt]{article}
\usepackage[margin=1in]{geometry}

\usepackage[utf8]{inputenc} % allow utf-8 input
\usepackage{csquotes}
\usepackage{booktabs}       % professional-quality tables
\usepackage{amsfonts,yfonts}       % blackboard math symbols
\usepackage{nicefrac}       % compact symbols for 1/2, etc.
\usepackage{microtype}      % microtypography
\usepackage{parskip}        % don't indent paragraphs 
\usepackage[dvipsnames]{xcolor} % way more colors
\usepackage{amsmath, amsthm, mathtools, dsfont}
\usepackage{txfonts} % notni
\usepackage{color,graphicx}
\usepackage{url,hyperref}
\usepackage{enumerate}
\usepackage{algorithm,,float}
\usepackage[noend]{algpseudocode}
\usepackage[most]{tcolorbox}
\usepackage{longfbox}
\usepackage{todonotes}
\usepackage{tikz}
\usepackage{tcolorbox}
\usepackage{framed}
\usepackage{mdframed}
\usepackage{wrapfig}
\usepackage{framed}

\usepackage[T1]{fontenc}

\colorlet{shadecolor}{gray!15}

\definecolor{ShadowColor}{RGB}{30,150,190}

\makeatletter
\newcommand\Cshadowbox{\VerbBox\@Cshadowbox}
\def\@Cshadowbox#1{%
  \setbox\@fancybox\hbox{\fbox{#1}}%
  \leavevmode\vbox{%
    \offinterlineskip
    \dimen@=\shadowsize
    \advance\dimen@ .5\fboxrule
    \hbox{\copy\@fancybox\kern.5\fboxrule\lower\shadowsize\hbox{%
      \color{ShadowColor}\vrule \@height\ht\@fancybox \@depth\dp\@fancybox \@width\dimen@}}%
    \vskip\dimexpr-\dimen@+0.5\fboxrule\relax
    \moveright\shadowsize\vbox{%
      \color{ShadowColor}\hrule \@width\wd\@fancybox \@height\dimen@}}}
\makeatother

\algdef{SE}[SUBALG]{Indent}{EndIndent}{}{\algorithmicend\ }%
\algtext*{Indent}
\algtext*{EndIndent}

\newcommand{\ceil}[1]{{\left\lceil #1 \right\rceil}}
\newcommand{\floor}[1]{\left\lfloor\,#1\,\right\rfloor}

\newcommand{\mI}{{\mathcal{I}}}

\newcommand{\REAL}{\ensuremath{\mathbb{R}}}

\renewcommand{\Pr}[1]{\ensuremath{\mathbb{P}\left[#1\right]}}

\newcommand{\Ex}[1]{\ensuremath{\mathbb{E}\left[#1\right]}}

\newcommand{\ind}[1]{\ensuremath{\mathds{1}\left[#1\right]}}

\newcommand{\poly}{\textnormal{poly}}
\DeclareMathOperator{\argmax}{arg\,max}
\DeclareMathOperator{\argmin}{arg\,max}
\newcommand{\eps}{\epsilon}
\newcommand{\abs}[1]{{\left|#1 \right|}}

\newcommand{\cbr}[1]{\left\{\,#1\,\right\}}

\newlength{\diamondrulelength}
\newlength{\diamondrulethickness}
\newcommand{\matroid}{(\mathcal{N}, \mI)}
\newcommand{\construct}{{\textsc{LayerConstruct}}}

\newcommand{\adjust}{{\textsc{Adjust}}}
\newcommand{\expand}{{\textsc{Expand}}}

\newcommand{\chbucket}{{\textsc{ChooseBucket}}}

\newcommand{\ins}{{\textsc{Insert}}}
\newcommand{\globins}{Global\ins}
\newcommand{\del}{{\textsc{Delete}}}
\newcommand{\globdel}{Global\del}
\newcommand{\init}{{\textsc{Init}}}
\newcommand{\inspect}{\textsc{Inspect}}
\newcommand{\buck}{B}
\newcommand{\estimatesize}{{\textsc{EstimateSampleSize}}}
\newcommand{\err}{\textsc{Fail}}

\newtheorem{theorem}{Theorem}
\newtheorem{corollary}[theorem]{Corollary}
\newtheorem{lemma}[theorem]{Lemma}
\newtheorem{fact}[theorem]{Fact}

\newtheorem{definition}[theorem]{Definition}

\newcommand{\relcont}[2]{f(#1 | #2)}
\newcommand{\bX}{\ensuremath{\mathbf{X}}}

\newcommand{\bQ}{\ensuremath{\mathbf{Q}}}

\newcommand{\bS}{\ensuremath{\mathbf{S}}}
\newcommand{\bs}{\ensuremath{\mathbf{s}}}
\newcommand{\sampvals}{\mathcal{S}}

\newcommand{\bP}{\ensuremath{\mathbf{P}}}

\newcommand{\bT}{\ensuremath{\mathbf{T}}}

\newcommand{\epsP}{\epsilon_{\textnormal{pr}}}
\newcommand{\epsDel}{\epsilon_{\textnormal{del}}}

\newcommand{\epsIns}{\epsilon_{\textnormal{ins}}}

\newcommand{\epsSamp}{\epsilon_{\textnormal{sam}}}
\newcommand{\epsBuck}{\epsilon}
\newcommand{\epsTau}{\epsilon}

\newcommand{\optInd}{I^*}

\newcommand{\constt}{{\frac{4}{\epsSamp^2}\ln(\frac{2nk^{10}}{\epsP})}}
\newcommand{\highprob}{1 - \frac{\epsP}{k^{10}}}
\newcommand{\ratio}{\frac{4k}{\eps}}
\newcommand{\invratio}{\frac{\eps}{4k}}
\newcommand{\suit}{\texttt{Suit}}
\newcommand{\trip}{(\indep, \uni, \delt)}
\newcommand{\acc}{X} % Used in estimate sample size
\newcommand{\add}{A}
\newcommand{\indep}{I}
\newcommand{\uni}{{I^{\cup}}}
\newcommand{\delt}{\Delta}
\newcommand{\elem}{E}
\newcommand{\celem}{\hat{\elem}}
\newcommand{\layer}{L}
\newcommand{\samp}{S}
\newcommand{\ssize}{s}

\newcommand{\setd}{D}
\newcommand{\hist}[1]{H_{< #1}}
\newcommand{\hists}[1]{H_{< #1}^{\text{s}}}
\newcommand{\bhist}[1]{\mathbf{H}_{< #1}}
\newcommand{\bhists}[1]{\mathbf{H}_{< #1}^{\text{s}}}
\newcommand{\spn}{\texttt{Span}}

\usepackage{environ}
\NewEnviron{killcontents}{Proof: Hidden}
\providetoggle{hideproofs}
\settoggle{hideproofs}{false}
\iftoggle{hideproofs}{

}

\title{From Linear to Polylogarithmic: \\ Dynamic Submodular Maximization over Matroids}

\begin{document}
% \author{}

\author{ 
Kiarash Banihashem\thanks{Computer Science Department, University of Maryland, College Park, MD, USA. {\tt kiarash@umd.edu}.}
\and
Samira Goudarzi \thanks{Computer Science Department, University of Maryland, College Park, MD, USA. {\tt samirag@umd.edu}. }
\and
MohammadTaghi Hajiaghayi\thanks{Computer Science Department, University of Maryland, College Park, MD, USA. {\tt hajiagha@cs.umd.edu}.}
\and
Peyman Jabbarzade\thanks{Computer Science Department, University of Maryland, College Park, MD, USA. {\tt peymanj@umd.edu}.}
\and
Morteza Monemizadeh\thanks{Department of Mathematics and Computer Science, TU Eindhoven, the Netherlands. {\tt m.monemizadeh@tue.nl}.}
}

\date{\vspace{-5ex}}

\maketitle

\begin{abstract}
We study the problem of submodular maximization under a matroid constraint. In this problem, there is a monotone submodular function $f: 2^{\mathcal{N}} \rightarrow \mathbb{R}^+$ and a matroid $\mathcal{M}(\mathcal{N}, \mathcal{I})$, accessible via value and independence oracles and 
the goal is to find the independent set $S \in \mathcal{I}$ maximizing the objective function, i.e., $f(S) = \max_{I \in \mathcal{I}} f(I)$. 
This problem has been extensively studied across a variety of settings, including the online setting~\cite{DBLP:journals/jacm/BabaioffIKK18} (JACM~2018),~\cite{DBLP:conf/stoc/KleinbergW12} (STOC~2012), the adaptive complexity model~\cite{DBLP:conf/stoc/BalkanskiRS19} (STOC~2019), and the streaming model~\cite{chakrabarti2015submodular} (Math. Program.~2015).

In the fully dynamic setting, the active ground set $N \subseteq \mathcal{N}$ evolves through element insertions and deletions. The goal is to maintain a solution $S \subseteq N$ with $S \in \mathcal{I}$ and $f(S) = \max_{I \subseteq N,, I \in \mathcal{I}} f(I)$ after each update. Two approximation algorithms have been proposed to solve this problem: Dütting et al.~\cite{DBLP:journals/talg/DuttingFLNZ25} developed a $(4+\epsilon)$-approximation algorithm with an expected amortized update time of $O\!\bigl(\frac{k^2}{\epsilon} \log k \log n \log^3 \frac{k}{\epsilon}\bigr)$, which grows quadratically with the matroid rank $k$, and Banihashem et al.~\cite{DBLP:conf/soda/BanihashemBGHJM24} achieved a similar approximation with a better update time $O\!\bigl(k \log k \log^3 \frac{k}{\epsilon}\bigr)$, which is linear in $k$.  
 However, these methods both become inefficient when the rank $k$ of the matroid is large (e.g., when $k = n^{o(1)}$). 

The central open question in this line of work is whether a fully dynamic algorithm with polylogarithmic update time exists. This question was posed in~\cite{DBLP:conf/nips/LattanziMNTZ20} (NeurIPS'20) and~\cite{DBLP:journals/talg/DuttingFLNZ25} (TALG'25), and remains a natural open problem following~\cite{DBLP:conf/stoc/ChenP22} (STOC'22) and~\cite{DBLP:conf/soda/BanihashemBGHJM24} (SODA'24). 

We resolve this question by presenting the first fully dynamic algorithm for submodular maximization over matroids with polylogarithmic update time. Our algorithm maintains a $(4+\epsilon)$-approximate solution with expected amortized update time $\mathrm{poly}(\log n, \epsilon^{-1})$, where $n = |\mathcal{N}|$. In particular, the update time is independent of the matroid rank $k$, thereby eliminating the dependence on $k$ and settling the open problem.
\end{abstract}

\section{Introduction}

Many Web applications repeatedly need to select a high-quality subset of items from a large and continuously changing collection. For example, a Web platform may need to select a set of articles, products, advertisements, or recommended items to display to users. The quality of a selected set is often naturally modeled by a submodular function: adding an item provides diminishing marginal benefit as increasingly similar or redundant items are already selected. Submodular objectives have been widely used to model coverage and information-gathering objectives~\cite{DBLP:journals/tist/KrauseG11}, as well as diversity in recommendation and search~\cite{pmlr-v51-tschiatschek16,WuEtAl2024}.

In practice, however, a platform cannot select an arbitrary subset of items. The selected items must satisfy combinatorial constraints reflecting limited display space, business rules, or diversity requirements. \emph{Matroid constraints} provide a particularly natural and expressive way to model such requirements. For example, a partition matroid can impose quotas on different categories or content providers, such as allowing at most $r_i$ items from category $i$. The \emph{rank} of the matroid then has a direct interpretation as the total number of available selection slots. More general matroids can capture overlapping eligibility and resource constraints while allowing the algorithm to exchange selected items in a structured way. Thus, submodular maximization under a matroid constraint provides a general framework for balancing the quality of a selection with diversity and resource limitations that arise naturally in Web applications.

The role of the matroid rank is particularly intuitive in recommendation and content-selection systems. Suppose a platform must display at most $k$ items. A uniform matroid of rank $k$ models the basic capacity constraint that at most $k$ items can be selected. A partition matroid can further distribute these $k$ slots across categories or providers, for example by requiring that no category contribute more than a specified number of items. The submodular objective can then favor selections that cover diverse topics, users, or interests. In this setting, the rank represents the scale of the selection that the platform can maintain, while the matroid structure specifies how those limited slots may be allocated. More sophisticated matroids can similarly represent multiple interacting resource or eligibility constraints.

A defining characteristic of many Web applications is that their underlying collections and interactions evolve continuously. For example, in recommender systems, new items are introduced and users continuously interact with existing items, changing the interaction graph and the information available for recommendation~\cite{TangEtAl2023,GuoEtAl2025}. Real-world recommendation systems must also adapt to changing item-popularity trends while scaling to millions of items~\cite{MaEtAl2020}. These changes can affect both the objective value of a selected set and its feasibility under application-specific constraints. Recomputing an approximately optimal solution from scratch after every insertion or deletion can be prohibitively expensive at Web scale. This motivates the fully dynamic version of submodular maximization, in which the active ground set changes over time and the algorithm must efficiently maintain a near-optimal feasible solution after each update. 

For instance, consider a recommendation platform that maintains a set of items to display while promoting diversity across categories or providers~\cite{WuEtAl2024}. The utility of a recommendation set can be modeled by a submodular function that captures user engagement, coverage, or diversity, while category or provider quotas can be represented by a matroid constraint. As new items become available and existing items are removed, the platform must continuously update its recommendations to adapt to the changing collection~\cite{GuoEtAl2025,TangEtAl2023}. Ideally, these updates should be processed in time that is polylogarithmic in the size of the system, rather than in time that grows with the number of available items or with the matroid rank.

Our work addresses this challenge by developing a fully dynamic algorithm for monotone submodular maximization under a matroid constraint with expected amortized update time $\mathrm{poly}(\log n,\epsilon^{-1})$, independent of the matroid rank. This result provides a general algorithmic framework for maintaining approximately optimal constrained selections in dynamic Web environments, where both the collection of available items and the feasible solution may change continuously.

Submodular functions naturally express the principle of \emph{diminishing returns}: 
the gain from adding an element decreases as the chosen set grows. 
Formally, a set function $f: 2^{\mathcal{N}} \rightarrow \mathbb{R}^+$ defined on a finite ground set $\mathcal{N}$ is 
\emph{submodular} if, for every $A \subseteq B \subseteq \mathcal{N}$ and element $e \in \mathcal{N} \setminus B$,  
$f(A \cup \{e\}) - f(A) \geq f(B \cup \{e\}) - f(B)$.  
If, in addition, $f(A) \le f(B)$ for all $A \subseteq B$, then $f$ is \emph{monotone}. 

Constraints often reflect structural or resource limitations—such as selecting a subset of features, 
sensors, or network links, which can be modeled as a matroid to capture independence and exchange properties. 
A \emph{matroid} $\mathcal{M} = (\mathcal{N}, \mathcal{I})$ is defined by a collection of independent sets 
$\mathcal{I} \subseteq 2^{\mathcal{N}}$ that satisfy non-emptiness, heredity, and the exchange property. 
The rank~$k$ of $\mathcal{M}$ is the size of its largest independent set. 
The classical offline problem of \emph{submodular maximization under a matroid constraint} seeks an independent set 
$S \in \mathcal{I}$ that maximizes $f(S)$, i.e., $\max_{S \in \mathcal{I}} f(S)$.

Submodular maximization has a broad range of applications. 
In machine learning, it appears in feature selection~\cite{DBLP:conf/icml/BateniCEFMR19,DBLP:conf/esa/ChenEFMY21}, 
data summarization~\cite{mirzasoleiman2016fast,DBLP:conf/kdd/BadanidiyuruMKK14}, 
and active learning~\cite{DBLP:conf/icml/WeiIB15,DBLP:conf/colt/GolovinK10}. 
In data mining and network science, it models 
influence maximization~\cite{DBLP:conf/kdd/KempeKT03,10.1145/1250790.1250811} and 
community detection~\cite{10.1145/3399661,10197512}, 
while in network design~\cite{DBLP:conf/icalp/ChekuriQ21a,DBLP:conf/soda/ChekuriQ19}, 
it captures objectives related to connectivity, reliability, and information flow.

The \emph{dynamic} variant of submodular maximization under a matroid constraint has received growing attention~\cite{DBLP:conf/soda/BanihashemBGHJM24,DBLP:journals/talg/DuttingFLNZ25,DBLP:journals/jea/HanauerHS22,DBLP:conf/sand/HanauerH022}. 
Here, the ground set $\mathcal{N}$ evolves over time through insertions and deletions, and 
the goal is to maintain an independent set $S \in \mathcal{I}$ that remains near-optimal for the current instance, 
while limiting the computational overhead per update. 
Formally, the algorithm interacts with $f$ through a value oracle and with $\mathcal{M}$ through an independence oracle, 
and must update its maintained solution as the input changes.

While the offline version of submodular maximization under a matroid constraint is well understood---for instance, the classic greedy algorithm attains a $(1 - 1/e)$-approximation~\cite{DBLP:journals/mp/NemhauserWF78}---these offline methods recompute the entire solution from scratch following each change, resulting in linear update time in $|\mathcal{N}|$. Such costs become prohibitive in large-scale settings, which has spurred interest in efficient \emph{dynamic submodular maximization algorithms} that offer near-optimal guarantees while maintaining polylogarithmic update times.

Designing efficient algorithms for maintaining near-optimal solutions to monotone submodular maximization under matroid constraints in a \emph{fully dynamic} setting has been a longstanding challenge (was posed in~\cite{DBLP:conf/nips/LattanziMNTZ20,DBLP:journals/talg/DuttingFLNZ25}, and remains a natural open problem following~\cite{DBLP:conf/stoc/ChenP22,DBLP:conf/soda/BanihashemBGHJM24}). 
The goal is to handle insertions and deletions of elements while simultaneously achieving two key objectives: 
\begin{itemize}
    \item maintain a constant-factor approximation (e.g., $(4+\epsilon)$) to the optimal solution, and
    \item ensure that each update requires only \emph{polylogarithmic} time in the size of the ground set, measured in oracle queries to $f$ and $\mathcal{M}$.
\end{itemize}

Prior work has made substantial progress toward this goal. 
D\"utting et al.~\cite{DBLP:journals/talg/DuttingFLNZ25} provided a $(4+\epsilon)$-approximation algorithm with expected amortized update time 
$O\bigl(\frac{k^2}{\epsilon} \log k \log n \log^3 \frac{k}{\epsilon}\bigr)$. 
In concurrent and independent work, Banihashem et al.~\cite{DBLP:conf/soda/BanihashemBGHJM24} improved this to a linear dependence on $k$, achieving 
$O\bigl(k \log k \log^3 \frac{k}{\epsilon}\bigr)$ expected amortized update time while preserving the same approximation guarantee. 
Despite these advances, both approaches become computationally expensive when the matroid rank is large, since the update time still grows polynomially with $k$.

In this work, we overcome this barrier by introducing the first fully dynamic algorithm for submodular maximization under a general matroid constraint whose update time is \emph{polylogarithmic} in $n$. 
Our algorithm maintains a $(4+\epsilon)$-approximate solution with expected amortized update time $\mathrm{poly}(\log n, \epsilon^{-1})$, independent of the matroid rank $k$. 
This breakthrough ensures that even rapidly evolving datasets can be handled efficiently, resolving the open problem posed by earlier studies. 
%and demonstrating that fully dynamic submodular maximization with matroid constraints can, in fact, scale gracefully with the size of the ground set. 

\paragraph{Overview of Our Algorithm.} 
Our dynamic algorithm for submodular maximization under matroid constraints builds on the dynamic framework developed for the simpler setting of cardinality constraints~\cite{DBLP:conf/icml/BanihashemBGHJM23}, which was later extended to submodular cover~\cite{banihashemdynamiccover} and submodular matching~\cite{DBLP:conf/icalp/BanihashemBGHJM25}.
Prior to our work, it remained unclear whether these ideas could yield an algorithm with polylogarithmic update time under matroid constraints, whose rich combinatorial structure is notoriously difficult to handle\footnote{
A well-known illustration of this difficulty is the classical \emph{secretary problem}.
Babaioff, Immorlica, and Kleinberg~\cite{DBLP:conf/soda/BabaioffIK07} introduced the matroid secretary problem as a generalization, in which the goal is to select a maximum-weight independent set of a matroid from elements arriving in random order.
Constant-competitive algorithms are known for several structured matroid families, including graphic, partition, and laminar matroids~\cite{DBLP:conf/icalp/KorulaP09,DBLP:conf/soda/ImW11,DBLP:journals/algorithmica/GharanV13}. Similarly, constant-factor approximations are known under matching constraints~\cite{ezra2020secretary}, suggesting that similar algorithms might also work for the matroid constraints and that comparable guarantees could extend to matroids.

However, achieving such guarantees for general matroids remains one of the central open problems in the area. Despite being conjectured as early as 2007, the existence of constant-competitive algorithms for general matroids is still unresolved after nearly two decades of research, and the best known results achieve an $O(\log\log k)$-competitive ratio, where $k$ denotes the matroid rank~\cite{DBLP:conf/focs/Lachish14,DBLP:conf/soda/FeldmanSZ15}.
This persistent gap highlights the inherent difficulty of extending techniques that succeed under uniform (cardinality) or matching constraints to the broader class of matroid constraints.
}.

%
% While this line of work was limited to a single cardinality bound, our approach extends to general matroid constraints through a layered structure combined with a monotone reference set for marginal gains, which allows us to maintain an approximate independent set throughout the dynamic process without recomputing from scratch after each update. 

To process elements, the algorithm organizes them into $T$ hierarchical layers, with $T$ polylogarithmic in expectation.
At each layer $\ell$, our data structure maintains an independent set $\indep_\ell \in \mathcal{I}$, together with an auxiliary set $\uni_\ell$ as the union of all elements present in the independent set $\indep$ of its preceding levels. A key design choice is that marginal gains are always evaluated with respect to $\uni$ rather than $\indep$. Since $\uni$ only grows over time, marginal values are non-increasing, and elements that were discarded at some layer because of their small marginal contribution do not need to be reconsidered. This prevents costly re-evaluations.

 Each layer $\ell$ contains a set of elements $E_\ell$ whose marginal gains with respect to $\uni$ exceed a threshold and are feasible for addition to $\indep$ under the matroid constraint. These elements are further partitioned into logarithmically many buckets according to their marginal values, and in each layer, the largest bucket is selected for processing. This bucket has two key properties: its elements have nearly equal marginal gains, and processing a constant fraction of them both ensures a substantial contribution to the solution and provides robustness against small deletions. These properties are crucial for controlling solution quality and for bounding the approximation guarantee.

Within each chosen bucket, we employ a sampling mechanism to estimate how many elements should be drawn so that a sufficient number can be successfully added to $\indep$ without violating independence or, if necessary, by displacing existing elements that are less valuable. Once the sample size is determined, a set $S_\ell$ of that size is drawn uniformly from the bucket. With high probability, most of these elements are successfully added to $\indep$. In the case of conflicts, elements currently in $\indep$ are removed to restore independence while remaining in $\uni$, and the total marginal gain of the inserted elements is always at least twice that of the evicted ones. This geometric dominance guarantees that the overall loss is bounded and that the approximation ratio is preserved over time.

Insertions and deletions are handled using buffered updates. For insertions, each layer maintains a buffer $A_\ell$ of new elements, and the layer is rebuilt only when the buffer reaches an $\epsilon$-fraction of the size of $E_\ell$. For deletions, a single shared buffer $D$ is used across all layers, and a layer is rebuilt once an $\epsilon$-fraction of its elements have been deleted. These lazy updates spread the cost of reconstruction evenly over time, leading to low amortized query complexity. 

Finally, to accommodate elements of widely varying submodular values, we run parallel copies of the algorithm with thresholds of the form $\tau = \frac{\epsilon}{4k} \cdot 2^i$, where $i$ ranges over integers. Each element participates only in the runs whose thresholds match its value interval, ensuring efficiency across the full value range. This thresholding technique is inspired by the parallel-run method of~\cite{DBLP:conf/nips/LattanziMNTZ20}. 

% In summary, the combination of a monotone reference set for marginal gains, a polylogarithmic layered bucketing structure, a sampling-based update mechanism with geometric dominance, and buffered insertions and deletions yields the first dynamic algorithm for submodular maximization under general matroid constraints with polylogarithmic amortized query complexity. 

% Our goal is to prove the following theorem in this section. 

\begin{theorem}\label{thm:main}
For any $\epsilon > 0$, there exists a fully dynamic randomized algorithm for submodular maximization under a matroid constraint that maintains an expected $(4+\epsilon)$-approximate solution against an oblivious adversary. The expected amortized update time is $\operatorname{poly}(\log n, \frac{1}{\epsilon})$, where $n = |\mathcal{N}|$ is the ground set size. 
\end{theorem}

Next, we highlight the key novelties of our work in comparison with prior dynamic algorithms for submodular maximization under matroid and cardinality constraints. 
The dynamic algorithm for matroid constraints proposed in~\cite{DBLP:conf/soda/BanihashemBGHJM24} achieves an update time linear in the rank~$k$ of the underlying matroid. 
In contrast, the dynamic algorithm for the cardinality constraint introduced in~\cite{DBLP:conf/icml/BanihashemBGHJM23} achieves a polylogarithmic update time. 
However, neither of these techniques alone can be extended to develop a dynamic algorithm for matroid constraints that maintains a $(4+\epsilon)$-approximate solution with an expected amortized update time of $\mathrm{poly}(\log n, \epsilon^{-1})$. 

%\paragraph{Comparison with~\cite{DBLP:conf/soda/BanihashemBGHJM24}.} 
The approach in~\cite{DBLP:conf/soda/BanihashemBGHJM24} maintains $O(k)$ layers and, upon the deletion of an element~$e$ from the current solution, reconstructs the layer~$\ell$ where~$e$ was added, along with all subsequent layers. 
As a result, this framework cannot be directly transformed into a data structure that supports $\mathrm{poly}(\log n, \epsilon^{-1})$ update time under matroid constraints. 
To overcome this limitation, we partition the active set~$E_{\ell}$ in each layer $\ell$ into logarithmically many buckets according to the marginal values of their elements, and we select the largest bucket for processing. 
This bucket contains at least $O(\frac{\epsilon}{\log n}|E_{\ell}|)$ elements whose marginal gains are nearly identical. 
Using a sampling mechanism, we estimate the fraction of elements in this bucket that can be added to the current solution through a single uniform sample. 
The key insight is that, since all elements added to the solution at layer~$\ell$ have nearly equal marginal gains, as long as only an $O(\epsilon)$ fraction of them are modified, the maintained solution continues to provide a $(4+\epsilon)$-approximation guarantee. 
Consequently, unlike the approach in~\cite{DBLP:conf/soda/BanihashemBGHJM24}, which can tolerate only one update to the maintained solution (and therefore requires $O(k)$ layers), our method can tolerate changes to an $O(\epsilon)$ fraction of elements with nearly equal marginal gains. 
This property enables us to achieve an expected amortized update time of $\mathrm{poly}(\log n, \epsilon^{-1})$.

%\paragraph{Comparison with~\cite{DBLP:conf/icml/BanihashemBGHJM23}.} 

% Moreover, although submodular maximization under a cardinality constraint can be viewed as a special case of the problem under a matroid constraint—where the matroid is uniform of rank~$k$—extending dynamic algorithms (e.g.,~\cite{DBLP:conf/nips/LattanziMNTZ20,DBLP:conf/icml/BanihashemBGHJM23}) from the cardinality case to general matroids is far from trivial. 
% The main difficulty lies in the structural nature of matroids—particularly the \emph{exchange property}, which dictates the feasibility of maintaining independence when elements are inserted or deleted. 
% Under a cardinality constraint, feasibility depends solely on the size of the selected set.
% This simplicity allows for efficient update operations and enables amortized polylogarithmic maintenance of approximate solutions. 
% In contrast, for general matroids, maintaining independence requires intricate coordination among elements, as the inclusion or removal of one element can affect the feasibility of many others. 
% Hence, algorithms designed for cardinality constraints cannot be directly extended to handle general matroids without substantial new ideas. 
% Our analysis requires non-trivial arguments that utilize matroid properties—particularly the exchange axiom—to ensure that updates maintain both independence and near-optimality. 

Moreover, while submodular maximization under cardinality or matching constraints is closely related to the matroid-constrained problem—specifically where the matroid is uniform of rank k—extending dynamic algorithms (e.g., \cite{DBLP:conf/nips/LattanziMNTZ20,DBLP:conf/icml/BanihashemBGHJM23,DBLP:conf/icalp/BanihashemBGHJM25}) to general matroids is far from trivial.
The primary challenge lies in the structural nature of matroids, particularly the strong exchange property, which governs how independence is maintained as elements are inserted or deleted. Under a cardinality constraint, feasibility is a simple function of the set's size. This simplicity facilitates efficient update operations and enables the maintenance of approximate solutions with amortized polylogarithmic complexity. Similarly, matching constraints possess a specific localized structure that limits the ripple effect of updates.
In contrast, for general matroids, maintaining independence requires intricate coordination across the entire ground set; adding or removing a single element can shift the feasibility status of numerous other subsets. Consequently, existing algorithms cannot be directly adapted to general matroids without substantial new insights. Our analysis introduces non-trivial arguments leveraging matroid axioms to ensure that dynamic updates preserve both independence and near-optimality.

One of the technicalities of this dynamic framework lies in its sampling size estimation procedure, which differs fundamentally from the binary-search sampling method used in~\cite{DBLP:conf/icml/BanihashemBGHJM23}, originally introduced in \cite{DBLP:conf/soda/FahrbachMZ19} and \cite{DBLP:conf/stoc/BalkanskiRS19}. 
Due to the exchange property inherent in matroid constraints, this binary-search approach cannot be directly applied. 
Instead, we design a new sampling mechanism adapted to the matroid structure. 
For each layer~$\ell$, to estimate the number of elements of the largest bucket that can be uniformly sampled all at once (i.e., non-adaptively), 
we consider a polylogarithmic number of random permutations of elements in the largest bucket. 
For each permutation, we iterate from its beginning to determine the elements that can be added to the current solution while maintaining independence. 
(Recall that if a conflict arises, elements currently in~$\indep$ are removed to restore feasibility, and the removed elements are added to~$\uni$. 
It is important to note that, unlike the dynamic algorithm for cardinality constraint~\cite{DBLP:conf/icml/BanihashemBGHJM23}, where marginal gains are compared directly against the current solution, in our matroid-constrained setting, the marginal gains are compared with respect to $\uni$, while the independence condition is enforced relative to the solution set $\indep$.)  

After processing all permutations, we identify the largest index $i$ of the permutations such that, in at least a $(1-\epsilon)$-fraction of the permutations, the element at index $i$ could be added to $\uni$. 
Finally, in contrast to the cardinality-constrained setting, where all $t$ uniformly sampled elements from the largest bucket are added to the solution, our framework must verify each sampled element individually to ensure matroid feasibility. 

However, the primary technical challenge of this work, distinguishing it from \cite{DBLP:conf/icalp/BanihashemBGHJM25}, arises from the fundamental structural disparities between matroids and matchings. Specifically, the \(\inspect\) and \(\adjust\) procedures required for matroids are more involved than those employed in the matching framework. These structural differences ultimately lead to a substantially different and more complex analysis to establish approximation guarantees.

\subsection{Related Work}

Dynamic submodular maximization has attracted significant attention in recent years. 
Lattanzi et al.~\cite{DBLP:conf/nips/LattanziMNTZ20} presented a randomized dynamic algorithm 
that maintains a $(2+\epsilon)$-approximate solution for maximizing a monotone submodular function under a cardinality constraint~$k$. 
Their method achieves an amortized expected query complexity  
$O(\epsilon^{-11} \log^6(k) \log^2(n))$ per update, which is polylogarithmic in~$n$. 
Independently, Monemizadeh~\cite{DBLP:conf/nips/Monemizadeh20} developed another randomized dynamic algorithm with similar approximation guarantees, requiring $O(\epsilon^{-3} k^2 \log^5(n))$ queries per update. 
% Extensions of these results to the dynamic non-monotone setting have also been developed~\cite{DBLP:conf/nips/BanihashemBGHJM23}.  
% Related work  further demonstrated that near-optimal efficiency is possible ~\cite{DBLP:conf/nips/LattanziMNTZ20,DBLP:conf/icml/BanihashemBGHJM23}, 
% \emph{submodular cover}~\cite{DBLP:conf/icml/BanihashemGHJM24}, maximum submodular matching~\cite{DBLP:conf/icalp/BanihashemBGHJM25} 
% in fully dynamic settings by achieving polylogarithmic query complexity for cardinality constraint, 
% But obtaining polylogarithmic update time for the the matroid constraint is widely open 
% despite persistent effort. We address this open problem in this paper. 
Subsequent work extended these results to the dynamic non-monotone setting~\cite{DBLP:conf/nips/BanihashemBGHJM23}. 
Further advances demonstrated that near-optimal efficiency is achievable for related problems such as \emph{submodular cover}~\cite{DBLP:conf/icml/BanihashemGHJM24} and \emph{maximum submodular matching}~\cite{DBLP:conf/icalp/BanihashemBGHJM25}, attaining polylogarithmic query complexity. 
However, achieving polylogarithmic update time for matroid constraints remains an open challenge despite persistent efforts~\cite{DBLP:conf/soda/BanihashemBGHJM24,DBLP:conf/stoc/ChenP22,DBLP:journals/talg/DuttingFLNZ25,DBLP:conf/nips/LattanziMNTZ20}. 
This paper resolves this open problem.

Chen and Peng~\cite{DBLP:conf/stoc/ChenP22} established lower bounds for dynamic submodular maximization with cardinality constraints, which is a special case of matroid constraints. 
They showed that any algorithm achieving better than a $(2-\epsilon)$-approximation must have amortized query complexity at least $n^{\tilde{\Omega}(\epsilon)} / k^3$, and achieving an approximation better than $1.712$ requires $\Omega(n / k^3)$ queries per update, even if the optimal value of the current set is known. 
In the insertion-only model, they further provided algorithms for submodular maximization under a matroid constraint of rank $k$, achieving $(2+\epsilon)$ and $(\frac{e}{e-1}+\epsilon)$-approximations with expected oracle query complexity $k^{\tilde{O}(1/\epsilon)}$ and $k^{\tilde{O}(1/\epsilon^2)} \cdot \log(n)$, respectively.  

For general matroid constraints, Dütting et al.~\cite{DBLP:journals/talg/DuttingFLNZ25} proposed a $(4+\epsilon)$-approximation algorithm with expected amortized update time $O\!\bigl(\frac{k^2}{\epsilon} \log k \log n \log^3 \tfrac{k}{\epsilon}\bigr)$, which scales quadratically with the matroid rank~$k$. 
An improved method was later introduced by Banihashem et al.~\cite{DBLP:conf/soda/BanihashemBGHJM24}, which achieves a similar approximation ratio but with expected amortized update time $O(k \log k \log^3 \tfrac{k}{\epsilon})$, linear in~$k$. 
Although these results represent progress, both approaches become inefficient when $k$ is large (e.g., $k = n^{o(1)}$), since their complexity still grows polynomially in~$k$.

% Polylogarithmic query time has also been achieved beyond cardinality and matroid constraints. 
% For example, the dynamic \emph{submodular cover} problem was studied in~\cite{DBLP:conf/icml/BanihashemGHJM24}, 
% where the goal is to find a subset $S$ such that $f(S)=f(\mathcal{N})$, and a dynamic algorithm with polylogarithmic query time per update was obtained. 
% More recently, Banihashem et al.~\cite{DBLP:conf/icalp/BanihashemBGHJM25} presented a dynamic $(8+\epsilon)$-approximation algorithm for the 
% \emph{maximum submodular matching} problem, with $\poly(\log n, \epsilon^{-1})$ amortized query time per update, where $n$ is the number of vertices in the graph.  

A related line of research considers the \emph{streaming} setting, where the data arrives sequentially and algorithms must work with limited memory. 
The first streaming algorithm for submodular maximization under a cardinality constraint was proposed by Badanidiyuru et al.~\cite{DBLP:conf/kdd/BadanidiyuruMKK14}, who gave a $(2+\epsilon)$-approximation using $O(k\epsilon^{-1} \log k)$ space. 
This was later improved by Kazemi et al.~\cite{DBLP:conf/icml/0001MZLK19}, who reduced the space requirement to $O(k\epsilon^{-1})$. 
Chakrabarti and Kale~\cite{chakrabarti2015submodular} introduced a streaming framework applicable to a wide range of combinatorial constraints, including matroid and matching constraints. 
Their approach yields a $4$-approximation algorithm for submodular maximization under a matroid constraint, using $O(k)$ space.

\subsection{Preliminaries}
\label{sec:prelimn}
In this section, we introduce submodular functions, matroids, the oblivious adversarial model used in our dynamic setting, and the notion of query complexity.

\paragraph{Submodular functions.}
Let $\mathcal{N}$ be a finite ground set. A function $f: 2^{\mathcal{N}} \rightarrow \REAL^+$ is called \emph{submodular} if for all $A \subseteq B \subseteq \mathcal{N}$ and any $e \notin B$,
\(f(A \cup \{e\}) - f(A) \ge f(B \cup \{e\}) - f(B),\)
capturing the principle of \emph{diminishing returns}: the benefit of adding an element $e$ to a smaller set $A$ is at least as large as adding it to its supersets.
An equivalent definition states that for all $A, B \subseteq \mathcal{N}$,
\(f(A) + f(B) \ge f(A \cup B) + f(A \cap B),\).
%highlighting that submodular functions are discrete analogs of concave functions.
We assume that the submodular function $f$ is also \emph{monotone}, meaning \(f(A \cup \{e\}) \ge f(A)\) for all \(A \subseteq \mathcal{N}\).
The value \(f(A \cup \{e\}) - f(A)\) is called the \emph{marginal gain} of adding $e$ to $A$, and is denoted by \(\relcont{e}{A}\).
In this paper, we assume to have access to \(f\) via a \emph{value oracle}.

\paragraph{Matroids.} \emph{Matroids} are general combinatorial structures that capture the notion of independence. Formally, a matroid $\mathcal{M} = (\mathcal{N}, \mathcal{I})$ consists of a finite ground set $\mathcal{N}$ and a collection of subsets $\mathcal{I} \subseteq 2^{\mathcal{N}}$, called \emph{independent sets}, satisfying the following three properties: \emph{non-emptiness:} $\emptyset \in \mathcal{I}$, i.e., the empty set is independent; \emph{heredity (downward-closedness):} If $A \in \mathcal{I}$ and $B \subseteq A$, then $B \in \mathcal{I}$, i.e., every subset of an independent set is also independent; \emph{exchange property:} If $A, B \in \mathcal{I}$ and $|A| < |B|$, then there exists an element $x \in B \setminus A$ such that $A \cup \{x\} \in \mathcal{I}$. The \emph{rank} $k$ of the matroid is the size of its largest independent set; the exchange property ensures that all maximal independent sets (called \emph{bases}) have the same size $k$. A \emph{circuit} in a matroid is a minimal dependent set, i.e., a set $C \subseteq \mathcal{N}$ that is not independent but every proper subset of $C$ is independent. Circuits describe the minimal sets of elements whose presence causes dependence; thus, removal of at least one element from each circuit of a dependent set makes it independent. The family $\mathcal{C}$ of circuits of a matroid $\mathcal{M}$ satisfies the following properties: $\emptyset \notin \mathcal{C}$; if $C_1, C_2 \in \mathcal{C}$ and $C_1 \subseteq C_2$, then $C_1 = C_2$; and if $C_1, C_2 \in \mathcal{C}$, $C_1 \neq C_2$ and $e \in C_1 \cap C_2$, then there exists $C_3 \in \mathcal{C}$ such that $C_3 \subseteq (C_1 \cup C_2) \setminus \{e\}$. 
It follows that if $I \in \mathcal{I}$ and $e \in \mathcal{N}$ with $I \cup \{e\} \notin \mathcal{I}$, then $I \cup \{e\}$ contains a unique circuit $C$ (see Lemma~\ref{lem:=1circuit}); consequently, removing any element of $C$ from $I \cup \{e\}$ restores its independence.
In this paper, we assume that we have access to an independence oracle for the matroid $\mathcal{M}$.

% A \emph{matroid} $\mathcal{M} = (\mathcal{N}, \mathcal{I})$ consists of a ground set $\mathcal{N}$ and a collection of subsets $\mathcal{I} \subseteq 2^{\mathcal{N}}$ called \emph{independent sets}, which satisfy three properties: non-emptiness, heredity (downward-closedness), and the exchange property. The rank $k$ of the matroid is the size of its largest independent set. The classical problem of \emph{submodular maximization under a matroid constraint} asks to find a set $S \in \mathcal{I}$ that maximizes $f(S)$, that is, to compute $\max_{S \in \mathcal{I}} f(S)$.

\paragraph{Oblivious adversarial model.}
We analyze our dynamic algorithm in the \emph{oblivious adversarial model}, a standard framework for randomized data structures~\cite{DBLP:conf/stoc/CarterW77}. In this model, the adversary knows both $f$ and the algorithm but not its random bits; thus, the entire update sequence is fixed before execution.

\paragraph{Query complexity.}
The \emph{query complexity} of a dynamic algorithm is the number of oracle queries (to the value oracle for $f$ and to the independence oracle for the matroid $\mathcal{M}$) used to maintain an approximate solution at time $t$, given all computation up to time $t-1$. We measure time complexity primarily in terms of this query complexity.

Finally, we give a few notations that we use frequently in this paper. 

For a function $x$ and a set $A$, we write $x[A]$ to denote the restriction of $x$ to the domain $A$. 
Given an event $E$, we use $\ind{E}$ to denote its \emph{indicator function}: it equals $1$ if $E$ occurs, and $0$ otherwise. 
We use bold letters for random variables (e.g., $\textbf{X}$) and the corresponding non-bold letters (e.g., $X$) to represent their realizations.  
The probability and expectation of a random variable $\textbf{X}$ are denoted by $\Pr{\textbf{X}}$ and $\Ex{\textbf{X}}$, respectively. 
For two events $A$ and $B$, we write $\Pr{A \mid B}$ for the conditional probability of $A$ given $B$.  
For an event $A$ with nonzero probability and a discrete random variable $\textbf{X}$, the conditional expectation of $\textbf{X}$ given $A$ is written as $\Ex{\textbf{X} \mid A}$ and defined by $\Ex{\textbf{X} \mid A} = \sum_{x} x \cdot \Pr{\textbf{X} = x \mid A}.$ 
Similarly, for two discrete random variables $\textbf{X}$ and $\textbf{Y}$, the conditional expectation of $\textbf{X}$ given that $\textbf{Y}=y$ is denoted by $\Ex{\textbf{X} \mid \textbf{Y}=y}$.

\section{Algorithms}
\label{sec:alg}
In this section, we describe our dynamic algorithm for submodular maximization under a matroid constraint that maintains an independent set~$\indep$ that achieves a~$(4+\eps)$-approximation to the optimal solution, 
while efficiently handling insertions and deletions. 
The algorithm is built using several algorithms or subroutines:

\paragraph{Inspection procedure (\inspect):}
For a new element~$e$, we first check whether its marginal gain with respect to current set~$\uni$ exceeds threshold~$\tau$.  
If~$e$ can be added to the independent set~$\indep$ without violating the matroid constraint, we include it.  
Otherwise, we identify an element~$\hat{e} \in \indep$ with the smallest marginal contribution~$\delt(\hat{e})$, 
such that replacing~$\hat{e}$ by~$e$ still preserves independence.  
If~$2 \cdot \delt(\hat{e}) \le \relcont{e}{\uni}$, we replace~$\hat{e}$ with~$e$.

\begin{algorithm}[H]
  \caption{Inspection}
  \begin{algorithmic}[1]
        
        \Procedure{\inspect}{$e, (\indep, \uni, \delt), \tau$}
        \label{alg:inspect}
        
          \If{$ \relcont{e}{\uni} < \tau$}\label{line:remove_w_not_in_range}
              \State \Return \err
          \EndIf
          \If{$\indep \cup\cbr{e} \in \mI$}
             \State \Return $\emptyset$
          \EndIf

          \State $\hat{e} \gets \arg\min \cbr{\delt(e') : e' \in \indep, 
          \indep \cup \cbr{e} \setminus \cbr{e'} \in \mI}$ \Comment{find by using a binary search}
          \label{line:matroid:promote:find_min}
          \If{$2\cdot \delt(\hat{e}) \le \relcont{e}{\uni} $} 
             \State \Return $\cbr{\hat{e}}$
          \EndIf
          \State \Return \err
        \EndProcedure

  \end{algorithmic}
\end{algorithm}

\paragraph{Iterative Expansion (\expand):}
Given a sample set~$\samp$ of candidate elements and the current triplet~$(\indep, \uni, \delt)$, 
we iterate over~$\samp$, applying \textproc{adjust} to each element~$e$.  
The procedure \textproc{adjust} tries to insert~$e$ via \textproc{inspect}, and if successful, updates marginal contributions and keeps~$\indep$ independent.

\begin{algorithm}[H]
  \caption{Iterative Expansion}
  \label{alg:expand}
  \begin{algorithmic}[1]

     \Procedure{\expand}{$\samp, (\indep, \uni, \delt), \tau$}
        \State $(\indep, \uni, \delt)' \gets (\indep, \uni, \delt)$
        \For{$e \in \samp$} %\label{line:construct_for}
            \State $(\indep, \uni, \delt)' \gets 
            \adjust (e, (\indep, \uni, \delt)', \tau)$
        \EndFor
        \State \Return $(\indep, \uni, \delt)'$ \label{line:construct_sample_process_end}
    \EndProcedure

    \Procedure{\adjust}{$e, (\indep, \uni, \delt), \tau$}
        \label{alg:adjust}
            \State $X \gets$ {\inspect}$(e, (\indep, \uni, \delt), \tau)$
            \If{$X \ne \err$}
                \State $\delt(e) \gets \relcont{e}{\uni}$ and $\indep \gets \indep \cup \cbr{e} \setminus X$ and $\uni{} \gets \uni \cup \cbr{e}$
            \EndIf
            \State \Return $(\indep, \uni, \delt)$
        \EndProcedure
  \end{algorithmic}
\end{algorithm}

\begin{algorithm}[H]
  \caption{Sample Size Estimation}
  \label{alg:sample_size}
  \begin{algorithmic}[1]
    \Procedure{\estimatesize}{$\buck, \trip, \tau$}
      \State $t \gets \ceil{\constt}$ 
      \State Initialize $X$ as a $t \times (|\buck| + 1)$ matrix of zeros.
      \For{$i \gets 1$ to $t$}
        \State $\buck' \gets \text{a random permutation of } \buck$
        \State $(\indep, \uni, \delt)' \gets \expand(\buck', \trip, \tau)$
        %\For{$j \gets 1$ to $|\buck'|$}
            \State $\acc[i,j] \gets \ind{\buck'[j] \in \uni{}'}$ \qquad {$\forall 1 \le j \le |\buck|$}
        %\EndFor
      \EndFor
      \State $ \ssize' \gets \min{\cbr{j : 1 \le j \le \abs{\buck} + 1, \frac{1}{t}(\sum_{i=1}^t \acc[i,j])< 1 - \epsSamp}} $
       %\State Let $\ssize'$ be the smallest index $j \in [1, |\buck| + 1]$ for which  $\frac{1}{t}(\sum_{i=1}^t \acc[i,j])< 1 - \epsSamp$
      \State \Return $\ssize' - 1$
    \EndProcedure
  \end{algorithmic}
\end{algorithm}

\paragraph{Sample Size Estimation (\estimatesize{}):}
At each layer~$\ell$, we partition active elements into buckets based on their marginal gains relative to~$\uni_{\ell-1}$.  
We estimate an appropriate sample size~$\ssize_{\ell}$ to ensure, with high probability, that sampling from the largest bucket yields enough elements 
that can be successfully added to~$\indep$.

\begin{algorithm}[H]
  \caption{Main Algorithm}
  \label{alg:main}
  \begin{algorithmic}[1]
    \Procedure{\init{}}{$\tau$}
    %\State $\elem_0, \add_0,  \setd, \indep_0, \uni_0 \gets \emptyset$
    \State $\indep_0, \uni_0, \elem_1, \add_1, \buck_1, \setd \gets \emptyset$ and $T \gets 0$
    \State Initialize $\delt_0$ as a zero function
    
    \EndProcedure

    \Procedure{\ins}{$e$}
      \State $D \gets D \setminus \{e\}$
      %\State $\elem_{0} \gets \elem_0 \cup \cbr{e}$
      \For{$\ell \gets 1$ to $T + 1$} \label{line:insert_loop}
        \If{$\inspect{}(e, \trip_{\ell - 1}, \tau) \ne \err$ \& $e \notin \elem_{\ell}$ }
            \State ${\add}_{\ell} \gets {\add}_{\ell}  \cup \cbr{e}$ 
        \EndIf
        \If{$\abs{\add_{\ell}} > \epsIns\abs{\elem_\ell}$}
        \label{line:insert_checkforconstruct}
          \State \Return $\construct{}(\ell)$
        \EndIf
      \EndFor
      \State \Return $\indep_T \setminus \setd$
    \EndProcedure

    \Procedure{\del}{$e$}
      \State $\setd \gets \setd \cup \cbr{e}$
      %\State $\elem_{0} \gets \elem_0 \setminus \cbr{e}$
      \For{$\ell \gets 1$ to $T$}
        \If{$\abs{\setd \cap \buck_\ell} \ge \epsDel \abs{\buck_\ell}$} \label{line:delete_cond}
          \State \Return $\construct{}(\ell)$
        \EndIf
      \EndFor
      \State \Return $\indep_T \setminus \setd$
        
    \EndProcedure    
    
    \Procedure{\construct}{$\ell$}
    %\textbf{Input:} A layer $\ell \ge 1$
      \State $\elem_{\ell} \gets (\elem_{\ell} \cup \add_{\ell}) \setminus \setd$
      \State $\add_{\ell} \gets \emptyset$
      \label{line:construct_init}

      \If{$\elem_{\ell} = \emptyset$}
        \State $T \gets \ell - 1$
        \State \Return $\indep_{T} \setminus \setd$
      \EndIf
      
        \State $(\buck_{\ell}, \tau_{\ell}) \gets \chbucket{}(\ell)$
        
        \State $\ssize_{\ell} \gets \estimatesize{}(\buck_{\ell}, \trip_{\ell-1}, \tau_{\ell})$
        \label{ln:estimatesz}
            
        \State $\samp_{\ell} \gets$ a random permutation of a uniformly selected subset of size $\ssize_{\ell}$ from $\buck_{\ell}$
        \label{line:construct_sample}
        
        \State $\trip_{\ell} \gets \expand(\samp_{\ell}, \trip_{\ell-1}, \tau_{\ell})$  
        \label{line:expand}
        
        \State $\elem_{\ell + 1} \gets \cbr{e \in \elem_{\ell}: \inspect{}(e, \trip_{\ell}, \tau) \ne \err}$
        \label{ln:filter}
        \State $\add_{\ell + 1} \gets \emptyset$
       
      \State \Return $\construct(\ell + 1)$
      \label{ln:ell+1}
    \EndProcedure

    \Procedure{\chbucket}{$\ell$}
    \State $\buck_{\ell, j} \gets \emptyset$ \quad $\forall  0 \le j 
    < \log_{(1+\epsTau)}(\ratio)$
    %\le \ceil{\log_{(1+\epsTau)} (\frac{\MAX{}}{\tauMin})}$
    \For{$e \in \elem_{\ell}$}
            \State $j \gets \floor{\log_{1+\eps}(\frac{\relcont{e}{\uni_{\ell - 1}}}{\tau})}$
            \State $\buck_{\ell, j} \gets \buck_{\ell, j} \cup \{e\}$
            \Comment{
            $\buck_{\ell, j} \gets \cbr{e \in \elem_{\ell}: (1+\eps)^{j}\tau \le \relcont{e}{\uni_{\ell - 1}} < (1+\eps)^{j+1}\tau}$
            }
        \EndFor
        
        \State $j_{\ell} \gets \argmax_{j} \cbr{\abs{\buck_{\ell, j}}}$

        \State \Return $(\buck_{\ell, j_{\ell}}, (1+\epsTau)^{j_{\ell}} \tau)$
        \Comment{
        $\buck_{\ell} \gets \buck_{\ell, j_{\ell}} $ \quad $\tau_{\ell} \gets (1+\epsTau)^{j_{\ell}} \tau $ 
        }
    
    \EndProcedure
  \end{algorithmic}
\end{algorithm}

\paragraph{Layered Dynamic Structure:}
Our algorithm organizes computation into multiple layers indexed by~$\ell$.  
For each layer, we maintain a set of active elements~$E_{\ell}$, an auxiliary buffer~$A_{\ell}$ for new insertions, and associated data structures.  
When a new element~$e$ arrives, we check in each layer whether it can be added to~$E_{\ell}$.  
If so, we perform a \emph{lazy update}: we add~$e$ to~$A_{\ell}$.  
When~$|A_{\ell}|$ reaches an~$\epsilon$-fraction of~$|E_{\ell}|$, we merge~$A_{\ell}$ into~$E_{\ell}$ and reconstruct layer~$\ell$ from scratch.

For deletions, we keep a single shared buffer~$D$ collecting all removed elements.  
For each layer~$\ell$, if an~$\epsilon$-fraction of~$E_{\ell}$ appear in~$D$, we rebuild that layer using the remaining undeleted elements.  
Unlike insertions—where each layer maintains its own buffer~$A_{\ell}$—deletions use the same global buffer~$D$.

\paragraph{Layer Construction (\construct):}
To build the next layer, the algorithm uses \chbucket{} to select the largest bucket of elements with similar marginal gains.  
We then sample a subset from this bucket, expand~$\indep$ by adding feasible elements (via \expand), and filter elements to create the new set~$A_{\ell+1}$ for the subsequent layer.

\paragraph{Parallel Instances:}
To handle elements with diverse marginal gains, we run parallel copies of the dynamic algorithm, each using a different threshold~$\tau = \frac{\eps}{4k} \cdot 2^i$.  
When an element~$e$ is inserted or deleted, we update all instances whose threshold ranges include~$f(e)$, following the parallel threshold idea from~\cite{DBLP:conf/nips/LattanziMNTZ20}.

\begin{algorithm}[H] 
\caption{Parallel Runs} 
\label{alg:parallel}
\begin{algorithmic}[1]
    \For{$i \in \mathbb{Z}$} 
        \State Let $\mathcal{I}_i$ be an instance of Main Algorithm \ref{alg:main}, with $\tau = \invratio \cdot 2^i$. 
    \EndFor

    \Function{\globins}{$e$}
    \label{alg:globins}
        \State $V \gets V \cup \cbr{e}$
        \For{\textbf{each} $
        \log(f(v))
        < i \leq \log(f(v)) + \log(\ratio)$
        } \label{line:paralle_run_insertion_interval}
        \Comment{
        \(\invratio \cdot 2^{i} \le f(v) < 2^{i}\)
        }
            \State Invoke $\ins(e)$ for instance $\mathcal{I}_i$.
        \EndFor
    \EndFunction

    \Function{\globdel}{$e$}
    \label{alg:globdel}
        \State $V \gets V \setminus \cbr{e}$
        \For{\textbf{each} $
        \log(f(v))
        < i \leq \log(f(v)) + \log(\ratio)$
        } \label{line:paralle_run_deletion_interval}
        \Comment{
        \(\invratio \cdot 2^{i} \le f(v) < 2^{i}\)
        }
            \State Invoke $\del(e)$ for instance $\mathcal{I}_i$.
        \EndFor
    \EndFunction

\end{algorithmic}
\end{algorithm}

\vspace{0.5em}
By combining layered reconstruction, careful sampling, and iterative expansion while always evaluating marginal gains relative to the persistent set~$\uni$,  
our algorithm dynamically maintains an independent set~$\indep$ that remains a~$(4+\eps)$-approximation to the optimum,  
even as elements are inserted to or removed from the ground set~$\mathcal{N}$.

\section{Analysis}
We consider a fixed point in time immediately after the algorithm finishes executing an update.
Throughout our analysis, we study the state of our algorithm at this fixed point of time, to which we may refer as the current time. As such, our Theorems, Definitions, and Lemmas are based on the current state of variables, unless explicitly specified.

\subsection{Deterministic Properties of the Data structure}
\label{sec:determin}
In this section, we establish a series of deterministic invariants that describe the state of our data structure at the said current time. We begin with a lemma that describes the structure of independent and union sets at layer zero. 

\begin{lemma}
\label{lm:initialII'}
The following holds:
\(
\indep_0 = \uni_0 = \emptyset.
\)
\end{lemma}

\begin{proof}
Consider the full course of the algorithm's execution.  
Both \(\indep_0\) and \(\uni_0\) are initialized as empty sets during the \(\init\) procedure, performed at the outset. Throughout the remainder of the algorithm, \(\indep_0\) and \(\uni_0\) remain unchanged, as none of the update procedures modify either of them, directly or indirectly, through the \(\construct\) procedure.  
This is because \(\construct(\ell)\) is only invoked for \(\ell \geq 1\) and does not make any modifications to \(\indep_{\ell'}\) or \(\uni_{\ell'}\) for any \(\ell' < \ell\).
\end{proof}

The following lemma establishes the relationship between the variables $\trip$ across successive layers, expressed in terms of the $\expand$ operation.
\begin{lemma}
\label{lm:invar:expand}
The following holds for all $\ell \in [T]$: 
\[
\trip_{\ell} = \expand(\samp_{\ell}, \trip_{\ell-1}, \tau_{\ell}).
\]
\end{lemma}
\begin{proof}
We prove this by considering a fixed $\ell \in [T]$ and showing that
\(
\trip_{\ell} = \expand(\samp_{\ell}, \trip_{\ell-1}, \tau_{\ell}).
\)

First, note that the helper functions \(\inspect\), \(\expand\), and \(\estimatesize\) do not modify their input parameters during execution.
Moreover, note that throughout the entire execution of the algorithm, the values of the variables \(\trip_i\), \(\samp_i\), and \(\tau_i\), for any arbitrary \(i\), may change only within the \(\construct(i)\) procedure prior to its invocation of \(\construct(i + 1)\) in line~\ref{ln:ell+1}.

Now, Consider the most recent invocation of $\construct(\ell)$. Immediately after the execution of line~\ref{line:expand} within $\construct(\ell)$, the equality
\(
\trip_{\ell} = \expand(\samp_{\ell}, \trip_{\ell-1}, \tau_{\ell})
\)
holds by construction.

Since any invocation of $\construct(\ell - 1)$ would subsequently trigger a call to $\construct(\ell)$, our choice implies that neither $\construct(\ell - 1)$ nor $\construct(\ell)$ has been invoked since then.
Hence, the values of $\trip_{\ell-1}$, $\samp_{\ell}$, $\tau_{\ell}$, and $\trip_{\ell}$ have remained unchanged since that point, confirming that the equality still holds at current time.
\end{proof}

The next lemma formalizes how the active elements are filtered across consecutive layers of our data structure. The proof of the following lemma captures how the algorithm handles reinsertions and redeletions that may occur during updates.
\begin{lemma}
\label{lm:invarFilter}
    Let \(\celem_{\ell}\) denote \((\elem_{\ell} \cup \add_{\ell}) \setminus \setd\) for \(\ell \in [T + 1]\).
    The following holds for all $\ell \in [T]$: 
    \[
    \celem_{\ell + 1} = \cbr{e \in \celem_{\ell}: \inspect(e, \trip_{\ell}, \tau) \ne \err}.
    \]
\end{lemma}

\begin{proof}
Recall that the statement of this lemma concerns the state of the variables following the execution of a fixed update, namely, the current time.

We prove this by showing that the statement holds with respect to the state of the variables before any update is applied, and that, assuming it holds prior to an update, we can conclude it also holds afterward. Hence, by induction, the statement holds for the current time.

Note that before any update is applied, the claim holds trivially as \(T = 0\).

Consider a call to the procedure \(\ins(v)\) or \(\del(v)\). For any variable \( X \), let \( X^{-} \) and \( X^{+} \) denote the state of \( X \) before and after the execution of the procedure, respectively.
Suppose that for all \( \ell \in [T^{-}] \),
    \[
    \celem_{\ell + 1}^{-} = \cbr{e \in \celem_{\ell}^{-}: \inspect(e, \trip_{\ell}^{-}, \tau) \ne \err}. 
    \]
We aim to prove that for all \( \ell \in [T^{+}] \),
\[
    \celem_{\ell + 1}^{+} = \cbr{e \in \celem_{\ell}^{+}: \inspect(e, \trip_{\ell}^{+}, \tau) \ne \err}. 
\]

We first prove the following fact, which will be used in our argument.
\begin{fact}
\label{fact5}
If \(\construct(\ell)\) is invoked during an update, and we use \( Y^{/} \) to denote the state of \( Y \) immediately before the invocation of \(\construct(\ell)\), then:
\[
\celem_{\ell}^{+} = \elem_{\ell}^{+} = \celem_{\ell}^{/}.
\]
\end{fact}

\begin{proof}
Note that \(\setd^{+} = \setd^{/}\). During the execution of \(\construct(\ell)\), the set \(\elem_{\ell}\) is updated as \((\elem_{\ell} \cup \add_{\ell}) \setminus \setd\) and does not change afterward. Then, \(\add_{\ell}\) is set to \(\emptyset\) and also remains unchanged. We have
\begin{align*}
\elem_{i}^{+} &= (\elem_{i}^{/} \cup \add_{i}^{/}) \setminus \setd^{/} = \celem_{i}^{/},\\
\add_{i}^{+} &= \emptyset,\\
\celem_{i}^{+} &= (\elem_{i}^{+} \cup \add_{i}^{+}) \setminus \setd^{+} =
(((\elem_{i}^{/} \cup \add_{i}^{/}) \setminus \setd^{/}) \cup \emptyset) \setminus \setd^{/}
=(\elem_{i}^{/} \cup \add_{i}^{/}) \setminus \setd^{/}.
\end{align*}
\end{proof}

Now, we proceed to prove that the following claim holds for an arbitrary fixed \( i \in [T^{+}] \), which will complete our proof.
\[
\celem_{i + 1}^{+} = \{ e \in \celem_{i}^{+} : \inspect(e, \trip_{i}^{+}, \tau) \ne \err \}.
\]

First, consider the case where \(\construct(i)\) is invoked.  
By Fact \ref{fact5}, we have 
\(
\celem_{i}^{+} = \elem_{i}^{+}.\)
When line~\ref{ln:filter} is executed, both \(\elem_{i}\) and \(\trip_{i}\) are finalized. Hence, \(\elem_{i + 1}\) is set as
\[
\{ e \in \elem_{i}^{+}: \inspect(e, \trip_{i}^{+}, \tau) \ne \err \} = \{ e \in \celem_{i}^{+} : \inspect(e, \trip_{i}^{+}, \tau) \ne \err \} .
\]
Then, \(\add_{i + 1}\) is set to \(\emptyset\). Since \(\add_{i + 1}\) is already \(\emptyset\), and \((\elem_{i + 1} \cap \setd) \subseteq (\elem_i \cap \setd) = \emptyset\), \(\elem_{i + 1}\) remains unchanged in \(\construct(i + 1)\). Hence, we have shown that
\[
\elem_{i + 1}^{+} = \{ e \in \elem_{i}^{+} = \celem_{i}^{+} : \inspect(e, \trip_{i}^{+}, \tau) \ne \err \}
\]
Moreover, by fact \ref{fact5}, we have 
\( \celem_{i + 1}^{+} = \elem_{i + 1}^{+}, \) thus proving the claim in this case.

Now, we consider the case where \(\construct(i)\) is not invoked.

Note that this assumption does not specify whether \(\construct(i + 1)\) is invoked.  
Let \( Y^{/} \) denote the state of the variable \( Y \) when we reach line~\ref{line:insert_checkforconstruct} in the case of executing \(\ins(v)\), or line~\ref{line:delete_cond} in the case of executing \(\del(v)\), with  \(\ell = i\).

We know that by that point, no change is made to \(\trip_{\ell}\) for \(\ell \leq i\), and all of \(\elem_{i}, \add_{i},\) and consequently \(\celem_{i}\) are already finalized; that is, 
\[
\elem_{i}^{+} = \elem_{i}^{/}\quad , \quad 
\add_{i}^{+} = \add_{i}^{/} \quad \text{and} \quad 
\celem_{i}^{+} = \celem_{i}^{/}.
\]

We also know that \(\elem_{i+1}\) and \(\add_{i+1}\) are already finalized, and that \(\celem_{i + 1}^{+} = \celem_{i + 1}^{/}\) unless \(\construct(i + 1)\) is invoked.  
Even though \(\elem_{i + 1}^{/}\) and \(\add_{i + 1}^{/}\) may then differ from \(\elem_{i + 1}^{+}\) and \(\add_{i + 1}^{+}\), by Fact~\ref{fact5}, we still have
\[
\celem_{i + 1}^{+} = \celem_{i + 1}^{/}.
\]
Therefore, regardless of the case, it suffices to show that
\begin{align}
\label{eq:pastnow}
    \celem_{i + 1}^{/} = \{ e \in \celem_{i}^{/} : \inspect(e, \trip_{i}^{-}, \tau) \neq \err \}, \end{align} given the assumption that
\begin{align}
\label{eq:pastassump}
\celem_{i + 1}^{-} = \{ e \in \celem_{i}^{-} : \inspect(e, \trip_{i}^{-}, \tau) \neq \err \}.
\end{align}

We prove this for \(\ins\) and \(\del\) separately.

1. During \(\ins(v)\), until we reach line~\ref{line:insert_checkforconstruct} with \(\ell = i\), 
the element \(v\) is removed from \(\setd\), and no change is made to \(\elem_{\ell}\) or \(\trip_{\ell - 1}\). 
For each \(\ell \in [i]\), the procedure checks whether \(v\) passes \(\inspect(v, \trip_{\ell - 1}, \tau)\).  
If it does, \(v\) is added to \(\add_{\ell}\).

Given these changes, we know that the pairs \(\celem_{i + 1}^{/}\) and \(\celem_{i + 1}^{-}\), and \(\celem_{i}^{/}\) and \(\celem_{i}^{-}\), may only differ at \(v\). 
Hence, to show that equation~\eqref{eq:pastnow} holds, we only need to focus on \(v\) and show that 
\begin{align}
\label{eq:ifonlyif}
v \in \celem_{i + 1}^{/} \text{ if and only if } 
v \in \celem_{i}^{/}, \text{ and } 
\inspect(v, \trip_{i}^{-}, \tau) \neq \err. 
\end{align}

We have that if \(v \in \celem_{i + 1}^{/}\), then \(v \in \elem_{i + 1}^{/} \cup \add_{i + 1}^{/} \). It follows that \(v \in \elem_{i + 1}^{-}\) or \(v \add_{i + 1}^{/}\). From both, we can infer that
\(\inspect(v, \trip_{i}^{-}, \tau) \neq \err.\)
It also follows that \(\inspect(v, \trip_{i - 1}^{-}, \tau) \neq \err.\) %todo by lemma
Hence, \(v\) was already in \(\elem_{i}^{-}\) or now is in 
\(\add_{i}^{/}\). Thus, \(v \in \celem_{i}^{/}.\)

Moreover, we have that if \(\inspect(v, \trip_{i}^{-}, \tau) \neq \err\), then \(v \in \add_{i + 1}^{/} \) if \(v \notin \elem_{i + 1}^{-}\). Therefore, \(v \in \add_{i + 1}^{/}\) or \(v \in \elem_{i + 1}^{/}\). Hence, the statement \eqref{eq:ifonlyif} holds. 

2. During \(\del(v)\), we only add \(v\) to \(\setd\) until we reach line \ref{line:delete_cond} with \(\ell = i\). Hence, we have 
\(\celem_{i}^{/} = \celem_{i}^{-} \setminus \{v\}\) and \(\celem_{i + 1}^{/} = \celem_{i + 1}^{-} \setminus \{v\}\). Therefore, the assumption \eqref{eq:pastassump} clearly implies equation \eqref{eq:pastnow}. 

This proves the claim in the case where \(\construct(i)\) is not invoked, thus completing the entire proof. 
\end{proof}

The next lemma establishes $T$ as the number of non-empty layers in our data structure. 
\begin{lemma}
\label{lm:inv:finish}
The following holds: 
\(\elem_{T + 1} = \emptyset\), and if \(T > 0\) then  \(\elem_{T} \neq \emptyset\). 
\end{lemma}

\begin{proof}
Consider the entire execution of the algorithm. The statement holds at the beginning, as initially \( T = 0 \) and \( \elem_{1} = \emptyset \).

Both update procedures, \(\ins(e)\) and \(\del(e)\), may change \( T \) and \( \elem_{T + 1} \) only through a call to the procedure \(\construct\).
Consider the last update that triggers a call to \(\construct\), and let \( i \) denote its corresponding layer number, i.e., \(\ins(e)\) or \(\del(e)\) has called \(\construct(i)\).
Note that if no such update occurs, then \( T \) and \( \elem_{1} \) have not changed since the beginning.

Invocation of \(\construct(i)\) implies that the condition in line \ref{line:insert_checkforconstruct} or \ref{line:delete_cond} did not hold for $1 \le j < i $, which guarantees \(\elem_{j} \neq \emptyset\) for all $1 \leq j < i$. 
Each call to \(\construct(\ell)\) is followed by a call to \(\construct(\ell + 1)\), continuing until reaching \(\construct(h)\) (\( h \geq i \)) with \( \elem_{h} = \emptyset \).
During \(\construct(h)\), the value of \( T \) is updated to \( h - 1 \), ensuring that \( \elem_{T + 1} = \emptyset \), while \( \elem_{j} \neq \emptyset \) for all \(1 \leq j \leq T\).
By our selection of the last update that triggers a call to \(\construct\), we conclude that no further changes have been made to \( T \) or \( \elem_{T + 1} \) thereafter.
\end{proof}

This final lemma characterizes the structural properties of our data structure imposed by the reconstruction conditions.
\begin{lemma}
\label{lm:reconstruction_condition}
The following hold:
\begin{align*}
    &|\buck_i \cap \setd| \leq \epsDel |\buck_i|, &\text{For all \( i \in [T] \),}\\
    &|\add_i| \leq \epsIns |\elem_i|.& \text{For all \( i \in [T + 1]. \)}
\end{align*}
\end{lemma}

\begin{proof}
Consider a fixed \( i \in [T] \).
Consider the most recent call to \(\construct(i)\). Immediately after its execution, we have \(\buck_i \cap D = \emptyset\). This holds because \(\elem_i\) was set to \((\elem_i \cup \add_i) \setminus \setd\), and \(\buck_i \subseteq \elem_i\).

Since any call to \(\construct(\ell)\) with \(\ell < i\) also triggers a call to \(\construct(i)\), our selection of the most recent call to \(\construct(i)\) implies that during all subsequent executions of \(\del(e)\), we reached line~\ref{line:delete_cond} with \(\ell = i\) without invoking \(\construct(i)\), meaning that the first condition was satisfied.
We also know that no subsequent \(\ins(e)\) could be the first update to violate condition~1, as insertions do not modify \(\buck_i\) or \(\setd\).

We make the same argument for insertions, but with a fixed  \( i \in [T + 1] \).
We consider the most recent call to \(\construct(i)\). Immediately after its execution, we have \(\add_i = \emptyset\).
During each subsequent insertion, we reached the line~\ref{line:insert_checkforconstruct} with \(\ell = i\), and \(\construct(i)\) was not invoked, meaning that the second condition was satisfied.
Moreover, no subsequent \(\del(v)\) could be the first update to violate condition~2, as deletions do not modify \(\elem_i\) or \(\add_i\).

Therefore, both conditions still hold. 
\end{proof}

\subsection{Approximation}
\label{sec:approx:first}
% This section initiates the analysis of the approximation ratio guaranteed by our algorithm. 
% To this end, we formalize the contribution of each element to the overall objective by introducing a weight function that quantifies the marginal significance of elements according to their state in the data structure.
This section initiates the analysis of the approximation ratio guaranteed by our algorithm. (Note that parts of this analysis draw inspiration from \cite{chakrabarti2015submodular}.) We begin by formalizing the contribution of each element to the overall objective by introducing a weight function that quantifies the marginal significance of elements according to their state in the data structure.
\begin{definition}
\label{def:weightfunction}
We define the weight function $w: \mathcal{N} \to \mathbb{R}$ as follows:
\begin{itemize}
  \item For $e \in \uni_T$, set $w(e) = \delt_T(e)$.
  \item For $e \notin \uni_T$, set $w(e) = \relcont{e}{\uni_T}$.
\end{itemize}
We also extend the definition of the weight function to sets by defining the total weight of a subset $S \subseteq \mathcal{N}$ as
$
w(S) := \sum_{e \in S} w(e),
$
with the convention that $w(\emptyset) := 0$ \footnote{Note that this definition is based on the simplifying assumption that \( f(\emptyset) = 0 \). Without this assumption, one should define \( w(S) := f(\emptyset) + \sum_{e \in S} w(e) \), with \( w(\emptyset) := f(\emptyset) \), to ensure the correctness of the lemmas.}.
\end{definition}

The following lemma demonstrates that when an element \(e\) gets accepted into, or already belongs to, the union set \(\uni\) during the execution of the \(\expand\) procedure, its associated marginal value \(\delt(e)\) remains fixed.
\begin{lemma}
\label{lm:w_withinlevel}
Consider the execution of \(\expand(\samp_{\ell}, \trip_{\ell, 0}, \tau_{\ell})\) with \(\tau_{\ell} > 0\). Let \( e_1, e_2, \ldots, e_{\abs{\samp_{\ell}}} \) denote the elements of \(\samp_{\ell}\) in their processing order. For any \(\eta \in [\abs{\samp_{\ell}}]\), let \(\trip_{\ell, \eta}\) denote the state of \(\trip'\) immediately after processing \(e_{\eta}\). If \(e \in \uni_{\ell, i}\) for some \(i \in [0, \abs{\samp_{\ell}}]\), then
\[
\delt_{\ell, i}(e) = \delt_{\ell, \abs{\samp_{\ell}}}(e).
\]

Informally, once \(e \in \uni\), the \(\expand\) procedure preserves the value of \(\delt(e)\) within that layer.
\end{lemma}
\begin{proof}
    First, observe that \(\expand(\samp_{\ell}, \trip_{\ell, 0}, \tau_{\ell})\) initializes \(\trip'\) as a copy of \(\trip_{\ell, 0}\), and that \(\trip_{\ell, \eta}\) denotes the state of \(\trip'\) immediately after processing \(e_{\eta}\). Therefore, for all \(\eta \in [\abs{\samp_{\ell}}]\), we have
    \begin{align}
        \label{thisequationhere}
        \trip_{\ell, \eta} = \adjust(e_{\eta}, \trip_{\ell, \eta - 1}, \tau_{\ell}).
    \end{align}

    Suppose, for the sake of contradiction, that \(\delt_{\ell, i}(e) \ne \delt_{\ell, \abs{\samp_{\ell}}}(e)\). Then there must exist some index \(j\) such that \(i < j \le \abs{\samp_{\ell}}\) and \(\delt_{\ell, j}(e) \ne \delt_{\ell, j - 1}(e)\).

    By \eqref{thisequationhere}, we have
    \[
    \trip_{\ell, j} = \adjust(e_j, \trip_{\ell, j - 1}, \tau_{\ell}).
    \]
    The procedure \(\adjust(e_j, \cdot, \tau_{\ell})\) may only change the value of \(\delt'\) at the element \(e_j\). Therefore, for \(\delt_{\ell, j}(e) \ne \delt_{\ell, j - 1}(e)\) to hold, it must be the case that \(e_j = e\).

    One may argue that given that \(e \in \uni_{\ell, i}\) and \(i < j\) this may never happen. However, we prove the Lemma without any further assumption on $\samp_{\ell}$.% and $\trip_{\ell, 0}$.

    We show that even if \(e_j = e\), the call \(\adjust(e, \trip_{\ell, j - 1}, \tau_{\ell})\) does not change the value of \(\delt'\) at element \(e\). This is because such an update requires that \(\inspect(e, \trip_{\ell, j - 1}, \tau_{\ell}) \ne \err\), which in turn requires:
\[
\relcont{e}{\uni_{\ell, j - 1}} \ge \tau_{\ell} > 0.
\]

However, \(e \in \uni_{\ell, i}\) and \(i < j\), implies \(\uni_{\ell, i} \subseteq \uni_{\ell, j - 1}\). Therefore, \(e \in \uni_{\ell, j - 1}\), and by definition, \(\relcont{e}{\uni_{\ell, j - 1}} = 0\).

This contradiction implies that \(\delt_{\ell, j}(e) = \delt_{\ell, j - 1}(e)\) for all \(j > i\), and hence \(\delt_{\ell, i}(e) = \delt_{\ell, \abs{\samp_{\ell}}}(e)\), as claimed.
\end{proof}

We now use the above lemma to relate the weight assigned to each element to its recorded \(\delt\) value across the different layers of the data structure, illustrating their alignment.
\begin{lemma}
\label{w=delt}
    If $e \in \uni_{i}$, then \(w(e) = \delt_i({e}) \). 
\end{lemma}
\begin{proof}
We show that given $e \in \uni_i$, we have $\delt_T(e) = \delt_i(e)$. Fix any index $h \in (i, T]$. By Lemma~\ref{lm:invar:expand}, we have:
\begin{align}
\label{eq:ex}
\trip_{h} = \expand(\samp_{h}, \trip_{h-1}, \tau_{h}).
\end{align}

Consider the execution of \(\expand(\samp_{h}, \trip_{h-1}, \tau_{h})\) as described in Lemma~\ref{lm:w_withinlevel}, by setting:
\[
\samp_{\ell} = \samp_{h}, \quad \trip_{\ell,0} = \trip_{h,0} = \trip_{h-1}, \quad \tau_{\ell} = \tau_{h}.
\]

Let \( e_1, e_2, \ldots, e_{\abs{\samp_{h}}} \) denote the elements of \(\samp_{h}\) in their processing order. For any $j \in [\abs{\samp_{h}}]$, let \(\trip_{h,j}\) denote the state of \(\trip'\) immediately after processing \(e_j\).  Lemma~\ref{lm:w_withinlevel} states that if \(e \in \uni_{h,q}\) for some $q \in [0, \abs{\samp_{h}}]$, then:
\[
\delt_{h,q}(e) = \delt_{h,\abs{\samp_{h}}}(e) = \delt_h(e).
\]

Since $h > i$, the condition \(e \in \uni_{h, q}\) holds for $q = 0$ as 
$\uni_{h, 0} = \uni_{h - 1}$, and by Lemma \ref{lm:unigrws} $e \in \uni_{i}$ implies $e \in \uni_{h - 1}$. Thus, by Lemma \ref{lm:w_withinlevel} we have: 
\[\delt_{h - 1}(e) = \delt_{h}(e).\]

By induction over $h$ from $i + 1$ to $T$, we conclude that $\delt_i(e) = \delt_T(e)$.

By Lemma~\ref{lm:unigrws}, we have $\uni_i \subseteq \uni_T$. Hence, $e \in \uni_i$ implies $e \in \uni_T$. By Definition~\ref{def:weightfunction}, the weight function for an element $e \in \uni_T$ is defined as \(w(e) = \delt_T(e)\). We also have established that given $e \in \uni_i$, $\delt_T(e) = \delt_i(e)$. Combined together, it follows that \(w(e) = \delt_i(e)\), as claimed.
\end{proof}

% \begin{fact}
%     If $e \notin \elem_{i}$, then \(w(e) \leq \)
% \end{fact}

We now proceed to derive bounds between the weight function or the objective function values corresponding to various sets maintained in our data structure and the optimal solution.

%The following lemma bounds the weight of the union set in terms of the weight of the independent set. 
\begin{lemma}
    \label{lm:wU<2wI}
    The following holds:
    $w(\uni_{T}) \le 2w(\indep_{T})$.
\end{lemma}
\begin{proof}
We prove by induction on $i$ that $w(\uni_i) \le 2w(\indep_i)$ for all $i \in [0, T]$.

For the base case ($i = 0$), we know by Lemma~\ref{lm:initialII'} that $\indep_0 = \uni_0 = \emptyset$.  
This, by Definition~\ref{def:weightfunction}, implies that $w(\uni_0) = w(\indep_0) = w(\emptyset)$, so the inequality holds trivially.

Now, for the inductive step, we show that $w(\uni_i) \le 2w(\indep_i)$ using the inductive hypothesis, which assumes the claim holds for $i - 1$, i.e., $w(\uni_{i-1}) \le 2w(\indep_{i-1})$. To this end, it suffices to show that
\[
w(\uni_i) - w(\uni_{i-1}) \le 2(w(\indep_i) - w(\indep_{i-1})).
\]

By Lemma~\ref{lm:invar:expand}, we have
\[
\trip_{i} = \expand(\samp_{i}, \trip_{i - 1}, \tau_{i}),
\]
where $\expand$ constructs $\trip_{i}$ by iteratively modifying a copy of $\trip_{i - 1}$ via calls to $\adjust$ for each $e \in \samp_{i}$.

Within $\adjust$, each element is sent through inspection, which can result in one of three possible outcomes.

If an element $e$ fails the test in $\inspect$, it is not added to either $\uni$ or $\indep$, and thus contributes nothing to either $w(\uni_i) - w(\uni_{i-1})$ or $w(\indep_i) - w(\indep_{i-1})$.  

Elements that pass the test in $\inspect$ and can be added to $\indep$ without any replacement are added to both $\uni$ and $\indep$, and thus contribute an equal amount of weight to both $w(\uni_i) - w(\uni_{i-1})$ and $w(\indep_i) - w(\indep_{i-1})$.  

Thus, it remains to examine the case where an element $e$ passes the $\inspect$ test and is added to both $\uni$ and $\indep$ while replacing another element $\hat{e}$ in $\indep$. 

It suffices to show that in such cases the following holds:
\[
w(e) \le 2 (w(e) - w(\hat{e})).
\]
% or equivalently,
% \(
% w(e) \ge 2w(\hat{e}). 
% \)

For such elements $e$ and $\hat{e}$, since both belong to $\uni_i$ and $\uni_i \subseteq \uni_T$ by Lemma \ref{lm:unigrws}, we have $e, \hat{e} \in \uni_T$. Hence, by Lemma \ref{w=delt}, it follows that $ w(e) = \delt_i({e})$ and $ w(\hat{e}) = \delt_i(\hat{e})$. Therefore, it suffices to show that \[\delt_i({e}) \ge 2\delt_i(\hat{e}).\]

We know that in such cases for the $\adjust$ procedure, during the $\inspect$ test, it is verified that $2\delt(\hat{e}) \le \relcont{e}{\uni}$.  Immediately after the inspection $\delt({e})$ is set to $\relcont{e}{\uni}$, and then $\indep$ and $\uni$ are modified accordingly. Once the $\adjust$ completes, the followings hold for the updated $\trip'$ within the $\expand$:
\[e, \hat{e} \in \uni', \quad \text{and} \quad  \delt'(e) \ge 2 \delt'(\hat{e}).\]

By Lemma \ref{lm:w_withinlevel}, we know that once an element is included in $\uni'$, its corresponding $\delt$ remains unchanged during the execution of $\expand$. Therefore, we have 
\[\delt_i({e}) \ge 2\delt_i(\hat{e}), \]
completing the proof.
\end{proof}

\begin{lemma} \label{lm:fU=wU}
The following holds: $ f(\uni_{T}) = w(\uni_{T})$.
\end{lemma}

\begin{proof}
Similar to Lemma~\ref{lm:wU<2wI}, we use an induction on \( i \), and prove that \( w(\uni_i) = f(\uni_i) \) for all \( i \in [0, T] \).

For the base case ($i = 0$), we know by Lemma~\ref{lm:initialII'} that $\uni_0 = \emptyset$, and by Definition \ref{def:weightfunction}, we have $w(\emptyset) = f(\emptyset)$. Thus, the base case holds.

For the inductive step, we assume that the claim holds for \( i - 1 \), i.e., \( w(\uni_{i-1}) = f(\uni_{i-1}) \), and prove that \( w(\uni_i) = f(\uni_i) \). It suffices to show that:
\begin{align}
\label{eq:stepwf}
w(\uni_i) - w(\uni_{i-1}) = f(\uni_i) - f(\uni_{i-1}).
\end{align}

By Lemma~\ref{lm:invar:expand}, we have
\[
\trip_i = \expand(\samp_i, \trip_{i-1}, \tau_i). 
\]
The procedure $\expand(\samp_i, \trip_{i-1}, \tau_i)$ initializes \( \trip' \) as a copy of \( \trip_{i-1} \), then sequentially modifies it via calls to $\adjust$ for each \( e \in \samp_i \), and uses it to set the value of \( \trip_i \) at the end.
Let \( e_1, e_2, \ldots, e_{\abs{\samp_{i}}} \) denote the elements of \( \samp_i \) in the order they are processed.
For each \( j \in [\abs{\samp_{i}}] \), let \( \uni_{i, j} \) denote the state of \( \uni' \) immediately after processing \( e_j \). 
This inherently implies that \( \uni_{i, \abs{\samp_{i}}} \) = \( \uni_{i} \). We also use \( \uni_{i, 0} \) to denote \( \uni_{i - 1} \).

Then, the equality~\eqref{eq:stepwf}, which we aim to prove can be rewritten as:
\begin{align}
\label{eq:telescoping}
\sum_{j=1}^{|\samp_i|} \left( w(\uni_{i,j}) - w(\uni_{i,j-1}) \right)
= \sum_{j=1}^{|\samp_i|} \left( f(\uni_{i,j}) - f(\uni_{i,j-1}) \right).
\end{align}

Now, we consider the outcome of the call to $\adjust$ on each element \( e_j \):

\begin{itemize}
    \item If \( e_j \) fails the test in $\inspect$, then \( \uni_{i,j} = \uni_{i,j-1} \), so the terms corresponding to \( e_j \) on both sides of \eqref{eq:telescoping} are equal to zero.
    
    \item Assuming \( e_j \) passes the $\inspect$ test, the following hold:
    \begin{itemize}
        \item[(a)] \( \uni_{i,j} = \uni_{i,j-1} \cup \{e_j\} \),
        \item[(b)] \( e_j \notin \uni_{i,j-1} \).
    \end{itemize}
    Property (b) follows because if \( e_j \in \uni_{i,j-1} \), then \( \relcont{e_j}{\uni_{i,j-1}} \) would be zero, contradicting the fact that \( e_j \) passed the \textsc{inspect} test, which requires \( \relcont{e_j}{\uni_{i,j-1}} \ge \tau_i \).
    
    From Definition~\ref{def:weightfunction}, together with (a) and (b), we have:
    \begin{align}
        \label{eq:wdelt}
        w(\uni_{i,j}) - w(\uni_{i,j-1}) = w(e_j) = \delt_i(e_j),
    \end{align}
    where the second equality follows from Lemma~\ref{w=delt} and the fact that \( e_j \in \uni_i \).
    
    Additionally, by the definition of relative contribution and using (a), we have:
    \begin{align}
        \label{eq:frelcon}
        f(\uni_{i,j}) - f(\uni_{i,j-1}) = \relcont{e_j}{\uni_{i,j-1}}.
    \end{align}
    
    Immediately after the completion of the $\adjust$ procedure for \( e_j \), we have:
    \[
    \delt'(e_j) = \relcont{e_j}{\uni_{i,j-1}}.
    \]
    By Lemma~\ref{lm:w_withinlevel}, we know that throughout the execution of $\expand$, once an element is included in $\uni'$, its corresponding \(\delt\) value remains unchanged. Hence,
    \[
    \delt_i(e_j) = \relcont{e_j}{\uni_{i,j-1}}.
    \]
    Therefore, using this together with Equations~\eqref{eq:wdelt} and~\eqref{eq:frelcon}, we conclude:
    \[
    w(\uni_{i,j}) - w(\uni_{i,j-1}) = f(\uni_{i,j}) - f(\uni_{i,j-1}).
    \]
\end{itemize}

Since, regardless of the case, each pair of corresponding terms in~\eqref{eq:telescoping} is equal, the equality holds. This completes the inductive step and the proof.
\end{proof}

\begin{lemma}
\label{lm:f_opt_w}
Let $\optInd$ denote the optimal solution at the current time. The following holds:
\[
f(\optInd) \le 2w(\indep_T) + w(\optInd \setminus \uni_T).
\]
\end{lemma}

\begin{proof}
We begin by bounding the difference between \( f(\optInd) \) and \( f(\uni_T) \). By submodularity of \( f \), we have:
\[
f(\optInd) - f(\uni_T) \le \sum_{e \in \optInd \setminus \uni_T} \relcont{e}{\uni_T}.
\]

By Definition~\ref{def:weightfunction},  for any element $e$ that \( e \notin \uni_T \), the weight function is defined as \( w(e) = \relcont{e}{\uni_T} \). Therefore, we have: 
\[
f(\optInd) \le f(\uni_T) + \sum_{e \in \optInd \setminus \uni_T} w(e) = f(\uni_T) + w(\optInd \setminus \uni_T), %\le f(\uni_T) + w(\optInd),
\]
where the equality follows from Definition~\ref{def:weightfunction} for sets, which defines the total weight of a set of elements as the sum of their individual weights.

By Lemma~\ref{lm:fU=wU}, we have \( f(\uni_T) = w(\uni_T) \), and by Lemma~\ref{lm:wU<2wI}, we know \( w(\uni_T) \le 2w(\indep_T) \). Therefore, we have:
\[
f(\uni_T) \le 2w(\indep_T).
\]

Combining the above, we conclude that:
\[
f(\optInd) \le f(\uni_T) + w(\optInd 
\setminus \uni_T) \le 2w(\indep_T) + w(\optInd 
\setminus \uni_T),
\]
as desired.
\end{proof}

The following lemma shows how the existence of a blocking subset in the independent set is preserved under the execution of the \(\adjust\) procedure for another element. 
\begin{lemma}
\label{lm:circ->circ}
Consider arbitrary elements \( e_1 \) and \( e_2 \), and a set \( \trip \). Assume that 
\(\relcont{e_1}{\uni} \ge \tau\), but 
\(\inspect(e_1, \trip, \tau) = \err\). Let \( \trip' \) be the output of 
\(\adjust(e_2, \trip, \tau)\). Then, there exists a set \( N \subseteq \indep' \) such that \( N \cup \cbr{e_1} \) forms a circuit, and 
\(\relcont{e_1}{\uni'} < 2\delt'(e')\) for all \( e' \in N \).
\end{lemma}

\begin{proof}
Firstly, we note a few helpful facts for the proof. 

1.We have \(\uni \subseteq \uni'\), and the submodularity of the function \(f\) ensures that \(\relcont{e_1}{\uni'} \le \relcont{e_1}{\uni}\).

2.\(\delt\) and \(\delt'\) may only differ at \(e_2\). 

3.The condition 
\(\relcont{e_1}{\uni'} < 2\delt(e')\) for all \( e' \in N \) is equivalent to 
\[
\relcont{e_1}{\uni'} < 2\delt'(\hat{e}) \text{ for }
\hat{e} = \arg\min \{ \delt'(e') : e' \in N 
\}.
\]

Since \(\relcont{e_1}{\uni} \ge \tau\), but \(\inspect(e_1, \trip, \tau) = \err\), \(\indep \cup \cbr{e_1}\) must be dependent.
We define \(C_1\) as:
\[C_1 := \cbr{e_1} \cup \cbr{e' \in \indep:
\indep \cup \cbr{e_1} \setminus \cbr{e'} \in \mI}. \]
We know that \(C_1\) is the unique  circuit in \(\indep \cup \cbr{e_1}\). (It can be easily verified that each element in \(C_1\) must be in a circuit, \(\indep \cup \cbr{e_1}\) contains only one circuit by Lemma \ref{lem:=1circuit}, and any element in the circuit is included in \(C_1\).) 
Since $\inspect(e_1, \trip, \tau) = \err$, for the element $\hat{e} \in C_1 \setminus \{e_1\}$ with minimum $\delt$, we have:
\[\relcont{e_1}{\uni} < 2 \delt(\hat{e}).\]
%It follows that \(C_1 \setminus \{e_1\}\) forms a circuit with.

Now, consider the $\adjust(e_2, \trip, \tau)$ procedure.  
The $\adjust$ procedure calls $\inspect(e_2, \trip, \tau)$. 
If the output is either $\err$ or \(\emptyset\) or \(\cbr{\hat{e}_2} \not\subseteq C_1\),
the set \( C_1 \setminus \{e_1\} \) satisfies all the requirements for \(N\). This is because in all these cases, \(C_1 \setminus \{e_1\} \subseteq \indep'\), \(C_1\) is a circuit and the facts imply that
\[
\relcont{e_1}{\uni'} \le \relcont{e_1}{\uni} < 2 \delt(\hat{e}) = 
2\delt'(\hat{e}) \text{ for } \hat{e} = \arg\min \{ \delt'(e') : e' \in  C_1 \setminus \{e_1\} \}.
\]

The only remaining case, is that if $\inspect(e_2, \trip, \tau)$ returns \(\{e_3\}\) that \(e_3 \in C_1\).
Let \(C_2\) denote the circuit in \(\indep \cup \cbr{e_2}\). We know that \(e_3 \) was the edge with minimum \(\delt\), and 
\[
\delt'(e_2) = \relcont{e_2}{\uni} \ge \delt(e_3) = \delt'(e_3).\]

Moreover, since \(e_3 \in C_1\), for all \(e' \in C_1 \cup C_2 \setminus \{e_1\}\), we have
\[\relcont{e_1}{\uni'} < 2\delt'(e'). \]

By Lemma \ref{lem:elimination:axiom}, there exists a circuit \(C_3\) such that
\[
C_3 \subseteq C_1 \cup C_2 \setminus \{e_3\}.
\]
We have established that \(C_3 \setminus \{e_1\}\) satisfies all the properties for \(N\). The proof is complete.

\end{proof}

Combining the previous results and leveraging our carefully designed Algorithm \ref{alg:parallel} alongside an elegant charging scheme yields a global bound on the optimum's value in terms of the weight of \(\indep_T\).
\begin{lemma}
\label{lm:f_I*<4I}
    Let $\optInd$ denote the optimal solution at the current time.
    % Assuming that
    % \(2^{i^{*} - 1} \le f(\optInd) < 2^{i^{*}}, \)
    There exists an instance $\mathcal{I}$ in Algorithm \ref{alg:parallel}, for which
    the following holds: 
    \[
    (1 - \frac{\eps}{2})f(\optInd) < 4 w(\indep_T).
    \]
\end{lemma}

\begin{proof}
    Let's assume that
    \(2^{i^{*} - 1} \le f(\optInd) < 2^{i^{*}}. \)
    In this Lemma, we bound \(w(\optInd \setminus \uni_T)\). 
    Define $\optInd_{\ge \tau}$ and $\optInd_{< \tau}$
    with \(\tau = \invratio \cdot 2^{i^*}\), 
    as follows: 
\begin{align*}
  \optInd_{\ge \tau} := \{e^* \in \optInd \setminus \uni_T : w(e^*) \ge \tau\},
  \quad
  \optInd_{< \tau} := \{e^* \in \optInd \setminus \uni_T : w(e^*) < \tau\}.
\end{align*}

By definition of weight function for set, we have
\begin{align*}
w(\optInd_{< \tau}) =\sum_{e\in \optInd_{< \tau}} w(e)
\le |\optInd_{< \tau}|\tau \le \tau k,
\end{align*}
where the last inequality follows from the fact that \(
\abs{\optInd_{< \tau}} \le \abs{\optInd} \le k. 
\) By assumption it follows that
\[\tau k \le \frac{\eps}{4} \cdot 2^{i^*} = \frac{\eps}{2} \cdot 2^{(i^* - 1)} \le \frac{\eps}{2}\cdot f(\optInd).
\] 
Hence, 
we have the following bound on \(w(\optInd_{< \tau})\) that 
\begin{align*}
\label{eq: w{< \tau}}
w(\optInd_{< \tau}) &\le \frac{\eps}{2} \cdot f(\optInd). 
\end{align*}

Now, we turn our focus to \(\optInd_{\ge \tau}\). 
By Definition \ref{def:weightfunction}, for \(e \notin \uni_T\), we have \(w(e) = \relcont{e}{\uni_T}\). Therefore, since for any element \(e \in \optInd_{\ge \tau}\), we know that \(e \notin \uni_T \), we have 
\[\tau \le w(e) = \relcont{e}{\uni_T} \le f(e), \]
where the last inequality comes from the submodularity of \(f\).
The assumption also implies that \(f(e) \le f(\optInd) < 2^{i^*}.\)
Consider the instance $\mathcal{I}_{i^*}$ in Algorithm \ref{alg:parallel}. We are certain that any element \(e \in \optInd_{\ge \tau}\) exists in the instance $\mathcal{I}_{i^*}$.

Consider a fixed arbitrary element \(e \in \optInd_{\ge \tau}\). 
For all \(\ell \in [T + 1]\), we define \(\celem_{\ell}\) as in Lemma \ref{lm:invarFilter} to be  
\((\elem_{\ell} \cup \add_{\ell}) \setminus \setd\).
We know that \(e\) is an element present at the current time, with \(f(e) \ge \tau\). Hence, it must exist in one of \(\elem_1\) or \(\add_1\), and \(e \notin \setd\). It follows that \(e \in \celem_1\). 

We know by Lemma \ref{lm:inv:finish} that \(\elem_{T + 1} = \emptyset\). Moreover, Lemma \ref{lm:reconstruction_condition} implies that \(\add_{T + 1} = \emptyset\). Hence,  
\(e \notin \celem_{T + 1}\). 

let \(\ell^*\) be a level such that \(e \in \celem_{\ell^*}\) but \(e \notin \celem_{\ell^* + 1}\). By Lemma \ref{lm:invarFilter}, we have 
\[
\celem_{\ell^* + 1} = \cbr{e' \in \celem_{\ell^*}: \inspect(e', \trip_{\ell^*}, \tau) \ne \err}.
\]

Therefore, we have 
\[
\inspect(e, \trip_{\ell^*}, \tau) = \err.
\]
Furthermore, since \(\relcont{e}{\uni_T} \ge \tau \), we have
\(\relcont{e}{\uni_{\ell^*}} \ge \tau. \)

We know from Lemma \ref{lm:invar:expand} that \(\trip_{T}\) is obtained from \(\trip_{\ell^*}\) by iteratively applying \(\adjust\) to the elements of \(\samp_{\ell^* + 1}\), then \(\samp_{\ell^* + 2}\), until \(\samp_{T}\).

Therefore, even when \(\ell^* \neq T\) by applying Lemma \ref{lm:circ->circ}, we know that there exists a set $N \subseteq \indep_{T}$ such that \(N \cup \cbr{e}\) forms a circuit, and for all \(e' \in N\), 
\[w(e) = \relcont{e}{\uni_{T}} < 2\delt_{T}(e') = 2w(e').\]
For each element \(e \in \optInd_{\ge \tau}\), we denote its corresponding \(N\) by \(N(e)\).

Consider a bipartite graph \((\optInd_{\ge \tau}, \indep_T, E) \) with an edge from \(e \in \optInd_{\ge \tau}\) to \(r \in  \indep_T\) if \(r \in N(e)\). 
For each \(e \in \optInd_{\ge \tau}\), we know that \(e \in \spn(N(e))\). 
This further implies that for all
\(E \subseteq \optInd_{\ge \tau}\), $E \subseteq \spn(N(E))$.
Therefore, Since $E$ is independent, we can
    %use Lemma~\ref{lm:matroid:hall_property_in_matroid} to 
conclude that $|E| \le |N(E)|$.
Thus, by Hall's theorem, There is a perfect matching \(M\) of \(\optInd_{\ge \tau}\) into \(\indep_T\). Since, weight of each element in \(\optInd_{\ge \tau}\) is less than twice the weight of its matched element in \(\indep_{T}\), we have 
\[w(\optInd_{\ge \tau}) < 2 w(\indep_T). \]

Therefore, we have 
\begin{align*}
    w(\optInd \setminus \uni_T) &= w(\optInd_{\ge \tau}) + w(\optInd_{< \tau})
     < 2 w(\indep_T) + \frac{\eps}{2} \cdot f(\optInd). 
\end{align*}
Then, by Lemma~\ref{lm:f_opt_w}, we obtain
\[(1 - \frac{\eps}{2})f(\optInd) < 4 w(\indep_T). \]
The proof is complete. 
\end{proof}

\subsection{Sample Size and Randomized Properties of the Data Structure}
\label{sec:random:sample:size}
This section analyzes the randomized aspects of our data structure. We begin by formally defining the notion of a suitable sample size, which is a key component of our algorithm.
\begin{definition}
\label{def:suitable_sample_size}
Consider given variables \(\buck\), \(\trip\), and \(\tau\).
Let \(\buck'\) be a random permutation of \(\buck\), and let \(\trip'\) denote the output of \( \expand(\buck', \trip, \tau)\). 
Define a random vector \(Y\) of length \( \abs{\buck} + 1 \) as follows: for each \(j \in [\abs{\buck}]\),
\[
Y[j] = \Pr{\buck'[j] \in \uni'},
\]
and set
\[
Y[\abs{\buck} + 1] = 0.
\]
We define \(\suit(\buck, \trip, \tau)\) as the set of indices \(s,\)
for which the following conditions hold:
\begin{enumerate}
    \item $Y[j] \ge 1 - \frac{3\epsSamp}{2} \quad \text{for all } j \le s, $
    \label{def:suitable-cond1}
    \item $Y[j] \le 1 - \frac{\epsSamp}{2} \quad \text{for } j = s + 1. $
    \label{def:suitable-cond2}
\end{enumerate}

In particular, for \(\ell \in [T]\), we define \(\suit_\ell := \suit(\buck_\ell, \trip_{\ell - 1}, \tau_\ell)\).  
\end{definition}

The following lemma further establishes the significance of the above definition by demonstrating desirable probabilistic properties of selecting a sample subset with a suitable size, guaranteeing both sufficient expansion of the solution at the current level and effective filtration of the remaining elements for the next level.
\begin{lemma}
    \label{lm:suit1impliessuit2}
    Consider \(\buck\), \(\trip\), \(\tau\), \(\buck'\), \(\trip'\), and \(Y\) as in Definition~\ref{def:suitable_sample_size}. Let $\samp$ be an ordered sample subset of \(\buck\) of size $\ssize$, and let \(\trip''\) denote the output of \( \expand(\samp, \trip, \tau)\). 
    
    Assume \(\ssize \in \suit(\buck, \trip, \tau)\), then the following hold:
    \begin{enumerate}
        \item \( \Ex{\abs{\samp \cap \uni''}} \ge (1- \frac{3\epsSamp}{2})\ssize, \)
        \label{lm:suitcon1}
        \item \(
        \Ex{
        \abs{\cbr{e \in \buck: \inspect{}(e, \trip'', \tau) = \err}}} \ge (\frac{\epsSamp}{2})\abs{\buck}.\)
        \label{lm:suitcon2}
    \end{enumerate}
\end{lemma}

\begin{proof}
    Recall that \(\buck'\) is a random permutation of \(\buck\), \(\trip'\) denotes the output of \( \expand(\buck', \trip, \tau)\), and \(Y\) is a random vector of length \( \abs{\buck} + 1 \) such that for all \(j \in [\abs{\buck}] \),
    \[
    Y[j] = \Pr{\buck'[j] \in \uni'}
    , \quad \text{and} \quad Y[\abs{\buck} + 1] = 0.
    \]

     For each \(j \in [\ssize]\), the random variables \(\ind{\buck'[j] \in \uni'}\) and \(\ind{\samp[j] \in \uni''}\) have the same distribution. 
    Therefore, 
    \begin{align*}
        \Ex{\abs{\samp \cap \uni''}} 
        &= \Ex{\sum_{j \in [\ssize]} \ind{\samp[j] \in \uni''}} 
        = \sum_{j \in [\ssize]} \Ex{\ind{\buck'[j] \in \uni'}}\\
        &= \sum_{j \in [\ssize]} \Pr{\buck'[j] \in \uni'}
        = \sum_{j \in [\ssize]} Y[j]
    \end{align*}
    
    Given the assumption that \(s \in \suit(\buck, \trip, \tau)\), condition \ref{def:suitable-cond1} in Definition \ref{def:suitable_sample_size} gives \(Y[j] \geq 1 - \frac{3\epsSamp}{2}\) for all \(j \leq \ssize\). Hence, 
    \[\Ex{\abs{\samp \cap \uni''}} \ge (1- \frac{3\epsSamp}{2})\ssize, \]
    as desired. 

    Next, consider \(\buck \setminus \samp\). If nonempty, let \(\overline{e}\) be an element selected uniformly at random from \(\buck \setminus \samp\), and define \(\overline{\samp}\) as \(\samp\) followed by \(\overline{e}\). Let \(\trip'''\) be \( \expand(\overline{\samp}, \trip, \tau)\). 
    By construction, 
    \[
    \expand(\overline{\samp}, \trip, \tau) = \adjust(\overline{e}, \trip'', \tau),
    \]
    and \(\overline{e} \in \uni'''\) if and only if \(\inspect(\overline{e}, \trip'', \tau) \neq \err\). Therefore,
    \[
    \Ex{\ind{\inspect(\overline{e}, \trip'', \tau) \neq \err}} = \Ex{\ind{\buck'[\ssize + 1] \in \uni'}} = Y[\ssize + 1] \leq 1 - \frac{\epsSamp}{2},
    \]
    where the first equality comes from the fact that the random variables \(\ind{\buck'[\ssize + 1] \in \uni'}\) and \(\ind{\overline{e} \in \uni'''}\) have the same distribution, and the last inequality uses condition (2) in Definition~\ref{def:suitable_sample_size}. It follows that
    \[
    \Ex{\ind{\inspect(\overline{e}, \trip'', \tau) = \err}} \geq \frac{\epsSamp}{2},
    \]
    and thus,
    \[
    \Ex{\abs{\{ e \in \buck \setminus \samp : \inspect(e, \trip'', \tau) = \err \}}} \geq \frac{\epsSamp}{2} \abs{\buck \setminus \samp}.
    \]

    Now, it remains to show that
    \[
    \Ex{\abs{\cbr{e \in \samp: \inspect{}(e, \trip'', \tau) = \err}}} \ge (\frac{\epsSamp}{2})\abs{\samp}.
    \]
    
    It is in fact true that \(\inspect(e, \trip'', \tau) = \err\) for any \(e \in \samp\). Yet, it suffices to show that this holds for any \(e \in \samp \cap \uni''\), which in turn implies 
    \[\abs{\cbr{e \in \samp: \inspect{}(e, \trip'', \tau) = \err}} \ge 
    \abs{\samp \cap \uni''}.\]
    This is because we have already established that \(\Ex{\abs{\samp \cap \uni''}} \ge (1- \tfrac{3\epsSamp}{2})\ssize, \) and \(1- \tfrac{3\epsSamp}{2} \ge \frac{\epsSamp}{2}\) under the assumption $\epsSamp \le \tfrac{1}{2}.$

    Finally, for all \(e \in \samp \cap \uni''\), we argue that \(\inspect(e, \trip'', \tau) = \err\) since  \(\inspect(e, \trip'', \tau) = \err\) if \(\relcont{e}{\uni''} < \tau\), which is the case for elements \(e \in \samp \cap \uni''\) because \(\relcont{e}{\uni''} = 0\), for \(e \in \uni''\). 

    The proof is complete.    
\end{proof}

The next lemma demonstrates that the procedure 
$\estimatesize$ returns a sample size from the set of suitable sample sizes defined above efficiently and with high probability.
\begin{lemma}
\label{lem:estimatesize_in_suit}
Suppose the procedure \(\estimatesize\) is called with input parameters \(\buck\), \(\trip\), and \(\tau\). 

With probability at least \(\highprob\), the following holds:
\[
\estimatesize(\buck, \trip, \tau) \in \suit(\buck, \trip, \tau).
\]
\end{lemma}

\begin{proof}
Hoeffding's inequality implies that for any \( j \in [|\buck|] \),
\begin{align*}
\Pr{ \left| \frac{\sum_{i=1}^t \acc[i, j]}{t} - \Ex{
\frac{\sum_{i=1}^t \acc[i, j]}{t}
} \right| \geq \frac{\epsSamp}{2} } \leq 2e^{-\frac{t \epsSamp^2}{4}}.
\end{align*}

Since we have:
\[
\acc[i, j] \coloneqq \ind{\buck'[j] \in \uni'}, \quad \text{and} \quad
Y[j] \coloneqq \Pr{\buck'[j] \in \uni'},
\]
it follows that
\[
\Ex{\frac{\sum_{i=1}^t \acc[i, j]}{t}} = Y[j].
\]

Thus, we have:
\begin{align*}
\Pr{\left| \frac{\sum_{i=1}^t \acc[i, j]}{t} - Y[j] \right| \geq \frac{\epsSamp}{2}} \leq 2e^{-\frac{t \epsSamp^2}{4}}.
\end{align*}

By substituting \( t = \constt \), we get:
\begin{align*}
\Pr{ \left| \frac{\sum_{i=1}^t \acc[i, j]}{t} - Y[j] \right| \geq \frac{\epsSamp}{2} }
\leq 2e^{-\ln\left( \frac{2nk^{10}}{\epsP} \right)} = \frac{\epsP}{nk^{10}}.
\end{align*}

By applying the union bound over all \( j \in [|\buck|] \), we have
\begin{align*}
\Pr{ \exists j \in [|\buck|] : \left| \frac{\sum_{i=1}^t \acc[i, j]}{t} - Y[j] \right| \geq \frac{\epsSamp}{2} }
&\leq |\buck| \cdot \frac{\epsP}{nk^{10}}
\leq \frac{\epsP}{k^{10}},
\end{align*}
where the last inequality uses \( |\buck| \leq n \). Therefore,
\begin{align*}
\Pr{ \forall j \in [|\buck|] : \left| \frac{\sum_{i=1}^t \acc[i, j]}{t} - Y[j] \right| < \frac{\epsSamp}{2} }
&\geq \highprob.
\end{align*}
Also note that  by definition, \(\left| \frac{\sum_{i=1}^t \acc[i, |\buck| + 1]}{t} - Y[|\buck| + 1] \right| = 0\). 
Hence, with probability at least \(\highprob\), for all \(j \in [|\buck| + 1]\), we have
\[
\frac{\sum_{i=1}^t \acc[i, j]}{t} - \frac{\epsSamp}{2} < Y[j] < \frac{\sum_{i=1}^t \acc[i, j]}{t} + \frac{\epsSamp}{2}. 
\]

We know that 
\(
\min{\cbr{j : 1 \le j \le \abs{\buck} + 1, \frac{\sum_{i=1}^t \acc[i,j]}{t}< 1 - \epsSamp}}
\) is selected as \(\ssize'\). 
Note that  since 
\( \frac{\sum_{i=1}^t \acc[i,|\buck| + 1]]}{t} = 0\), 
such \(\ssize'\) exists.
The choice of \(\ssize'\) implies that
\begin{align*}
    \frac{\sum_{i=1}^t \acc[i,j]}{t} &\geq 1 - \epsSamp &\text{for all }
    j < \ssize',\\
    \frac{\sum_{i=1}^t \acc[i,j]}{t} &< 1 - \epsSamp &\text{for } j = \ssize'.
\end{align*}

Hence, the following holds:
\begin{align*}
    Y[j] &> 1 - \frac{3\epsSamp}{2} &\text{for all } j < \ssize' - 1, \\
    Y[j] &< 1 - \frac{\epsSamp}{2} &\text{for } j = (\ssize' - 1) + 1.  
\end{align*}
This implies that
\[
\estimatesize(\buck, \trip, \tau) = \ssize' - 1 \in \suit(\buck, \trip, \tau).
\]
\end{proof}

Next, before proceeding with the randomized invariants of our dynamic algorithm, we first present a fact used in their analysis.
 \begin{fact}
 \label{fact}
Let \( C \) be collection of random variables \( X \), and events \( Y \), and let \( Z \) be a random variable. If \( Z \) is a function of \( X \), given \( Y \), that is, it can be deterministically determined given \( X \) and \( Y \), then:
\[
\Ex{A \mid X, Y, Z} = \Ex{A \mid X, Y} \quad \text{or} \quad \Pr{A \mid X, Y, Z} = \mathbb{P}[A \mid X, Y].
\]

The same holds if \( Z \) is an event that is implied by \( X \) and \( Y \).

However, note that if \( X \) and \( Y \) are contradictory to \( Z \), \( \Pr{A \mid X, Y, Z} \) and \( \Ex{A \mid X, Y, Z} \) are undefined.
Note that when we get \( \Ex{A \mid X, Y} = \Ex{\Ex{A \mid X, Y, Z}} \) by the law of iterated expectations, \( \Ex{A \mid X, Y, Z} \) is itself a random variable, a function only defined over values of \( (x, z) \) satisfying \( \Pr{\ X = x, Y, Z = z} > 0\).

 \end{fact}

We then follow with defining the notion of History \(\hist{\ell}\), which serves as the conditioning context for our randomized invariant concerning the sample size at each layer.

\begin{definition}
    Let \(\layer_{\ell}\) denote \((\trip_{\ell}, \elem_{\ell + 1}, \add_{\ell + 1})\) for \(\ell \in [0, T]\). 
    We define History \(\hist{\ell}\) as
    \begin{align*}
        \hist{\ell} &:= \cbr{\setd, \layer_{0}, \dots, \layer_{\ell - 1}}.  
    \end{align*}
    We use \(\bhist{\ell}\) to denote the random variables corresponding to \(\hist{\ell}\). 
    
    Note that conditioning on the values of \(\bhist{\ell}\) inherently implies \(\bT > \ell\).
\end{definition}

The following lemma establishes that, conditioned on its preceding history, the sample size at each layer is suitable with high probability. 
\begin{lemma}
\label{lm:invar:highprob}
For all $\ell \in [T]$, the following holds: 
\begin{align*}
\Pr{\bs_\ell \in  \suit_\ell \mid \bhist{\ell}} \ge \highprob.
\end{align*}
\end{lemma}

\newcommand{\rebuild}{C}
\begin{proof}
    We prove this by showing that the statement holds before any update is applied and that, assuming that it holds prior to an update, we can conclude that it also holds afterward.

    At the beginning, before any update, the claim holds trivially, as $T = 0$. Since the set $[T]$ is empty, the lemma statement holds vacuously for the initial state.
    
    For any variable \( X \), we use the notations \( X^{-} \) and \( X^{+} \) to denote \( X \) before and after the execution of an update, respectively. Hence, we aim to proof \(\Pr{\bs_\ell^{+} \in  \suit_\ell^{+} \mid \bhist{\ell}^{+}} \ge \highprob\) given the assumption that \(\Pr{\bs_\ell^{-} \in  \suit_\ell^{-} \mid \bhist{\ell}^{-}} \ge \highprob.\)
    
  Let $\rebuild_\ell$ be a random variable that takes the value $1$ if
  $\construct{}(\ell)$ was called during this update (whether directly or through $\construct{}(\ell - 1)$).
  
  By the law of iterated expectation, we have 
  \begin{align*}
    \Pr{\bs_\ell^{+} \in  \suit_\ell^{+} \mid \bhist{\ell}^{+}} = 
    \Ex{\Pr{\bs_\ell^{+} \in  \suit_\ell^{+} \mid \bhist{\ell}^{+}, \rebuild_\ell}}.
  \end{align*}
  Therefore, proving the following inequalities completes the proof. 
  \begin{align}
  \Pr{\bs_\ell^{+} \in  \suit_\ell^{+} \mid \bhist{\ell}^{+}, \rebuild_\ell=0} \ge \highprob,
    \label{eq:apr1_0945}
  \end{align}
  and
  \begin{align}
    \Pr{\bs_\ell^{+} \in  \suit_\ell^{+} \mid \bhist{\ell}^{+}, \rebuild_\ell=1} \ge \highprob.
    \label{eq:jul11_1815}
  \end{align}

  Equation~\eqref{eq:jul11_1815} holds by Lemma \ref{lem:estimatesize_in_suit}. 
  We also know that if $\bT^- < \ell$ but $\bT^+ \ge \ell$, then $\construct{}(\ell)$ has been called and therefore $\Pr{\rebuild_\ell=1} = 1$, which means Equation~\eqref{eq:jul11_1815} is sufficient to prove the claim for those levels. Thus we only need to focus on proving Equation \eqref{eq:apr1_0945} when $\bT^- \ge \ell$.

We know that conditioning on $\bhist{\ell}^{+}, \rebuild_\ell=0$ implies that  
  $\bT^- \ge \ell$. Thus, we can also condition on $\bhist{\ell}^{-}$.
    By the law of iterated expectation, we have 
    \begin{align*}
    \Pr{\bs_\ell^{+} \in  \suit_\ell^{+} \mid \bhist{\ell}^{+}, \rebuild_\ell=0}
    =
    \Ex{\Pr{\bs_\ell^{+} \in  \suit_\ell^{+} \mid \bhist{\ell}^{+}, \rebuild_\ell=0, \bhist{\ell}^{-}}}
  \end{align*}

  It therefore suffices to show that
  \begin{align*}
  \Pr{\bs_\ell^{+} \in  \suit_\ell^{+} \mid \bhist{\ell}^{+}, \rebuild_\ell=0, \bhist{\ell}^-} \ge \highprob. 
  \end{align*}

    Now we note that $\bs_\ell$ and $\suit_\ell$ would not change during an update unless the update triggers the \(\construct(\ell)\). Hence, given the condition $\rebuild_\ell=0$, we can replace \(\bs_\ell^{+} \in  \suit_\ell^{+} \) with 
    \(\bs_\ell^{-} \in  \suit_\ell^{-}\). This makes it sufficient to show that  
    \begin{align}
    \label{ineqdfhirgho}
  \Pr{\bs_\ell^{-} \in  \suit_\ell^{-} \mid \bhist{\ell}^{+}, \rebuild_\ell=0, \bhist{\ell}^-} \ge \highprob. 
  \end{align}

  Then we highlight that \(\bhist{\ell}^{+}\) is a function of \(\bhist{\ell}^{-}\), given $\rebuild_\ell=0$, and by Fact \ref{fact} we have
  \begin{align}
  \label{eqref:sfioweors}
  \Pr{\bs_\ell^{-} \in  \suit_\ell^{-} \mid \bhist{\ell}^{+}, \rebuild_\ell=0, \bhist{\ell}^-} = \Pr{\bs_\ell^{-} \in  \suit_\ell^{-} \mid \rebuild_\ell=0, \bhist{\ell}^-}. 
  \end{align}

    Lastly, we argue that $\rebuild_\ell$ is a function of $\bhist{\ell}^-$. In other words, $\bhist{\ell}^- = \hist{\ell}^-$ either implies $\rebuild_\ell=0$ or is in contradiction with it, and \(\Pr{\bs_\ell^{-} \in  \suit_\ell^{-} \mid \rebuild_\ell=0, \bhist{\ell}^-}\) is only defined for the former. Therefore, we also have 
    \begin{align}
    \label{eqref:suidfsofos}
  \Pr{\bs_\ell^{-} \in  \suit_\ell^{-} \mid \rebuild_\ell=0, \bhist{\ell}^-} = 
  \Pr{\bs_\ell^{-} \in  \suit_\ell^{-} \mid \bhist{\ell}^-} . 
  \end{align}

  Equations \eqref{eqref:suidfsofos} and \eqref{eqref:sfioweors} and our assumption prove the inequality \eqref{ineqdfhirgho}, and thus complete the proof.
\end{proof}

We now extend the history notion to define \(\hists{\ell}\),  incorporating the estimated sample size at layer \(\ell\) along with the history of preceding layers, used to condition the final randomized invariant of our data structure.
\begin{definition}
    Recall that \(\layer_{\ell}\) denotes \((\trip_{\ell}, \elem_{\ell + 1}, \add_{\ell + 1})\) for \(\ell \in [0, T]\). 
    Post-estimate history \(\hists{\ell}\) as
    \begin{align*}
    \hists{\ell} &:= \cbr{\hist{\ell}, \ssize} = \cbr{\setd, \layer_{0}, \dots, \layer_{\ell - 1}, \ssize_{\ell}}.  
    \end{align*}
    We let \(\bhists{\ell}\) denote the random variables corresponding to \(\hists{\ell}\). 
    
    Note that conditioning on the values of \(\bhists{\ell}\) inherently implies \(\bT > \ell\).
\end{definition}

With \(\bhists{\ell}\) defined, we now establish the uniformity of the sampled subset of each layer, conditioned on the post-estimate history of that layer. The proof of this lemma is similar to the proof of Lemma \ref{lm:invar:highprob}. 
\begin{lemma}
\label{lm:invar:uniform}
For all $\ell \in [T]$, let $\sampvals_\ell$ denote the
set of all ordered subsets of size $\ssize_\ell$ with elements of $\buck_{\ell}$. \\
For all \(\ell \in [T]\), the following holds:
\begin{align*}
\Pr{\bS_\ell = \samp \mid \bhists{\ell} = \hists{\ell}}
= \frac{\ind{\samp \in \sampvals_\ell}}{\abs{\sampvals_\ell}}
.
\end{align*}
\end{lemma}

\subsection{Effect of deletions on Approximation}
\label{sec:effect:del}
Since the algorithm returns $\indep_T \setminus D$ (more precisely, the best $\indep_T \setminus D$ among its instances), it suffices to bound $f(\optInd)$ by $f(\indep_T \setminus D)$. Having already established a connection between $f(\optInd)$ and $w(\indep_T)$, this section analyzes the relationship between $w(\indep_T)$ and $f(\indep_T \setminus D)$ and thereby finishes the approximation analysis. The proofs of some of the following lemmas rely on the randomized components of our data structure to capture the impact of removing $\setd$ from the final output.

The following Lemma bounds the weight of elements added to \(\uni\) at layer \(\ell \). 
\begin{lemma}\label{lm:quality_sample_count}
  For any $\ell \in [T]$, the following holds:
  \begin{align*}
    \Ex{w(\uni_\ell \backslash \uni_{\ell-1}) \mid \bhists{\ell} = \hists{\ell} }\ge 
    \Bigl(1- \frac{3\epsSamp}{2}\Bigr)
    \tau_\ell\ssize_\ell\ind{\ssize_\ell \in \suit_\ell}.
  \end{align*}
\end{lemma}

\begin{proof}
We consider a fixed $\ell \in [T]$ and a fixed condition $\bhists{\ell} = \hists{\ell}$.

Since the weight $w(\uni_\ell \backslash \uni_{\ell-1})$ is non-negative, by the law of total expectation we have
\begin{align*}
\Ex{w(\uni_\ell \backslash \uni_{\ell-1}) \mid \bhists{\ell} = \hists{\ell} } &\ge \Ex{w(\uni_\ell \backslash \uni_{\ell-1}) \mid \bhists{\ell} = \hists{\ell}, \ssize_\ell \in \suit_\ell}\ind{\ssize_\ell \in \suit_\ell}.
\end{align*}

Recall that by Lemma~\ref{lm:invar:expand}, we have
\[
\trip_{\ell} = \expand(\samp_{\ell}, \trip_{\ell-1}, \tau_{\ell}),
\]
where during $\expand$ each sampled element is processed by $\adjust$, and those passing the $\inspect$ test are added to $\uni_{\ell}$. Thus, the newly added elements $\uni_\ell \setminus \uni_{\ell-1}$ are those in $\samp_\ell$ that pass the $\inspect$ test, also equal to $\samp_{\ell} \cap \uni_\ell$. 

The $\inspect$ test guarantees that the relative contribution of any element $e$ added to $\uni$ satisfies $\delt_\ell(e) \ge \tau_\ell$. Moreover, by Lemma~\ref{w=delt}, we have $w(e) = \delt_\ell(e)$ for any such element. Therefore, for any newly added element $e$, $w(e) \ge \tau_\ell$.

Furthermore, by construction, $w(\uni_\ell \backslash \uni_{\ell-1})$ is the sum of the weights of the elements in $\uni_\ell \backslash \uni_{\ell-1}$.
We therefore have
\begin{align*}
\Ex{w(\uni_\ell \backslash \uni_{\ell-1}) \mid \bhists{\ell} = \hists{\ell}, \ssize_\ell \in \suit_\ell} &= \Ex{\sum_{e \in \uni_\ell \setminus \uni_{\ell-1}} w(e) \mid \bhists{\ell}= \hists{\ell}, \ssize_\ell \in \suit_\ell} \\
&\ge \Ex{\sum_{e \in \samp_{\ell} \cap \uni_\ell} \tau_\ell \mid \bhists{\ell}= \hists{\ell}, \ssize_\ell \in \suit_\ell} \\
&= \tau_\ell \Ex{\abs{\samp_{\ell} \cap \uni_\ell} \mid \bhists{\ell}= \hists{\ell}, \ssize_\ell \in \suit_\ell}.
\end{align*}

Since $\samp_\ell$ is an ordered sample of $\ssize_\ell$ elements from $\buck_\ell$, the condition $\ssize_{\ell} \in \suit_\ell$ and Lemma~\ref{lm:suit1impliessuit2} give 
\begin{align*}
 \Ex{\abs{\samp_{\ell} \cap \uni_\ell} \mid \bhists{\ell} = \hists{\ell}, \ssize_\ell \in \suit_\ell} \ge (1- \frac{3\epsSamp}{2})\ssize_{\ell}, 
\end{align*}
Combining the above with the initial inequality yields the desired bound
\begin{align*}
\Ex{w(\uni_\ell \backslash \uni_{\ell-1}) \mid \bhists{\ell} = \hists{\ell} }&\ge (1- \frac{3\epsSamp}{2}) \tau_\ell\ssize_\ell\ind{\ssize_\ell \in \suit_\ell}.
\end{align*}
\end{proof}

The following lemma bounds the total weight of elements added to $\uni$ at layer $\ell$ that belong to $\setd$, thereby quantifying the expected loss in the weight of $\uni$ at this layer, due to the deletions.
\begin{lemma}\label{lm:quality_deleted}
  For any $\ell \in [T]$, the following holds:
  \begin{align*}
    \Ex{w((\uni_\ell \backslash \uni_{\ell-1}) \cap D ) \mid \bhists{\ell} = \hists{\ell} }\le
    \epsDel (1+\epsBuck) \tau_\ell \ssize_\ell.
  \end{align*}
\end{lemma}

\begin{proof}
We consider a fixed $\ell \in [T]$ and condition on $\bhists{\ell} = \hists{\ell}$.
By Lemma~\ref{lm:invar:expand}, we have
\[
\trip_{\ell} = \expand(\bS_\ell, \trip_{\ell-1}, \tau_{\ell}),
\]
%Recall that $\bS_\ell$ is a subset of $\buck_\ell$ of size $\ssize_\ell$, and 
which means $\uni_\ell \setminus \uni_{\ell-1}$ is the set of elements in $\bS_\ell$ that are added to $\uni_{\ell}$ during the $\expand$ procedure.  
Thus, $\uni_\ell \setminus \uni_{\ell-1} \subseteq \bS_\ell$, and therefore
\[
\abs{(\uni_\ell \setminus \uni_{\ell-1}) \cap \setd} 
\le \abs{\bS_\ell \cap \setd} 
\]

By Lemma \ref{lm:invar:uniform}, we know that \(\bS_\ell\) is a uniform sample of size $\ssize_\ell$ from \(\buck_\ell\), and by Lemma \ref{lm:reconstruction_condition}, we have
\(
\abs{\buck_{\ell} \cap \setd} \leq \epsDel \abs{\buck_{\ell}}.
\) 
Therefore, we have
\[
\Ex{\abs{(\uni_\ell \setminus \uni_{\ell-1}) \cap \setd}}
\le \Ex{\abs{\bS_\ell \cap \setd} } 
= \frac{\ssize_{\ell}}{\abs{\buck_{\ell}}} \cdot \abs{\buck_{\ell} \cap \setd}
\le \frac{\ssize_{\ell}}{\abs{\buck_{\ell}}} \cdot \epsDel \abs{\buck_{\ell}}
\le  \epsDel \ssize_{\ell}.
\]

Furthermore, we know by Lemma~\ref{w=delt} that $w(e) = \delt_\ell(e)$ for any element in $\uni_\ell$. On the other hand, for any element $e \in \buck_i$, we know that $\relcont{e}{\uni_{\ell - 1}} \le (1+\epsBuck)\tau_\ell $. Also, by submodularity we know that the marginal contribution of element $e$ at the time of its addition to $\uni_{\ell}$ is less than or equal to its contribution with respect to $\uni_{\ell - 1}$, i.e., $\delt_\ell(e) \le \relcont{e}{\uni_{\ell - 1}}$. Therefore, for any element $e \in (\uni_\ell \backslash \uni_{\ell-1})$, we have 
\[w(e) \le (1+\epsBuck)\tau_\ell.\]

By construction, 
$w((\uni_\ell \setminus \uni_{\ell-1}) \cap \setd)$ 
is the sum of the weights of all elements in $(\uni_\ell \setminus \uni_{\ell-1}) \cap \setd$.  
Combining the above, we obtain
\begin{align*}
w((\uni_\ell \setminus \uni_{\ell-1}) \cap \setd) 
&= \sum_{e \in (\uni_\ell \setminus \uni_{\ell-1}) \cap \setd} w(e) \\
&\le \sum_{e \in (\uni_\ell \setminus \uni_{\ell-1}) \cap \setd} (1+\epsBuck)\tau_\ell \\
&= (1+\epsBuck)\tau_\ell \cdot \abs{(\uni_\ell \setminus \uni_{\ell-1}) \cap \setd}
\end{align*}

Combining the above proves that conditioned on $\bhists{\ell} = \hists{\ell}$,
\[
\Ex{w((\uni_\ell \setminus \uni_{\ell-1}) \cap \setd)}
\le \epsDel (1+\epsBuck) \tau_\ell \ssize_\ell. 
\]

\end{proof}

Building on the previous results, the following lemma demonstrates that $\uni_T$ retains almost all of its weight after the removal of $\setd$.
\begin{lemma}\label{lm:approx:w_U_D__w_U}
The following holds:
  \[
  \Ex{{w}({\uni_T} \setminus \setd)} \ge (1-\eps_2) \Ex{{w}({\uni_T})}, 
  \]
where 
  $\eps_2 := \frac{\epsDel(1+\epsBuck)}{1 - \tfrac{3}{2}\epsSamp} + \frac{8\epsP}{k^{8}\eps}$.
\end{lemma}

\begin{proof}
\newcommand{\bw}{\ensuremath{\mathbf{w}}}
\newcommand{\bd}{\ensuremath{\mathbf{d}}}
For each $\ell$, define the respective gain and loss of level $\ell$ as 
\[
\bw_\ell := w(\uni_\ell \setminus \uni_{\ell-1}), \quad 
\bd_\ell := w((\uni_\ell \setminus \uni_{\ell-1}) \cap \setd),
\]
with the convention that $\bw_\ell = \bd_\ell = 0$ for $\ell > T$.  
It follows from the definition that 
\[
w(\uni_T) = \sum_{\ell=1}^{\infty} \bw_\ell, 
\quad
w(\uni_T \setminus \setd) = \sum_{\ell=1}^{\infty} (\bw_\ell - \bd_\ell). 
\]

We aim to show that for each $\ell$
\begin{align}
\Ex{\bd_\ell \mid \bhist{\ell}} \le \eps_2 \Ex{\bw_\ell \mid \bhist{\ell}}.
\label{eq:target}
\end{align}
By the law of iterated expectation, this implies 
\[
\Ex{\bd_\ell} \le \eps_2 \Ex{\bw_\ell},
\]
and therefore
\[
\Ex{\bw_\ell - \bd_\ell} \ge (1-\eps_2)\Ex{\bw_\ell}.
\]
Finally, summing over $\ell$ and applying linearity of expectation yields
\[
\Ex{w(\uni_T \setminus \setd)} \ge (1-\eps_2)\Ex{w(\uni_T)},
\]
which completes the proof.

To establish Equation~\eqref{eq:target}, 
%we compare the expectations of $\bw_\ell$ and $\bd_\ell$ conditioned on $\bhists{\ell} = \hists{\ell}$, for any $\hists{\ell}$ with $\Pr{\bhists{\ell} = \hists{\ell}} > 0$. 
it suffices to show that
\begin{align}
\label{eq:modiftarg}
\Ex{\bd_\ell \mid \bhists{\ell} = \hists{\ell}} 
\le \eps_2 \, \Ex{\bw_\ell \mid \bhists{\ell} = \hists{\ell}},
\end{align}
for any $\hists{\ell}$ with $\Pr{\bhists{\ell} = \hists{\ell}} > 0$.

By Lemma~\ref{lm:quality_sample_count}, we have
\[
\Ex{\bw_\ell \mid \bhists{\ell} = \hists{\ell}} 
\ge \Bigl(1- \tfrac{3\epsSamp}{2}\Bigr)\ssize_\ell \tau_\ell \ind{\ssize_\ell \in \suit_\ell}.
\]
On the other hand, by Lemma~\ref{lm:quality_deleted}, we have
\[
\Ex{\bd_\ell \mid \bhists{\ell} = \hists{\ell}} 
\le \epsDel (1+\epsBuck)\tau_\ell \ssize_\ell.
\]
Combining the above gives us 
\[
\ind{\ssize_\ell \in \suit_\ell} \Ex{\bd_\ell \mid \bhists{\ell} = \hists{\ell}}
\le 
\frac{\epsDel (1+\epsBuck)}{1 - \tfrac{3\epsSamp}{2}}\;
\Ex{\bw_\ell \mid \bhists{\ell} = \hists{\ell}}, 
\]
which then implies
\[
\Ex{\bd_\ell \mid \bhists{\ell} = \hists{\ell}, \ssize_\ell \in \suit_\ell}
\le 
\frac{\epsDel (1+\epsBuck)}{1 - \tfrac{3\epsSamp}{2}}\;
\Ex{\bw_\ell \mid \bhists{\ell} = \hists{\ell}}.
\]

% In cases where $s_\ell \in \suit_\ell$, by combining the above, we obtain
% \[
% \Ex{\bd_\ell \mid \bhists{\ell} = \hists{\ell}, s_\ell \in \suit_\ell} 
% \le 
% \frac{\epsDel(1+\epsBuck)}{1 - \tfrac{3}{2}\epsSamp} \,
% \Ex{\bw_\ell \mid \bhists{\ell} = \hists{\ell}, s_\ell \in \suit_\ell}.
% \]

By Lemma \ref{lm:invar:highprob}, we have 
\(
 \Pr{\bs_\ell \in  \suit_\ell \mid \bhist{\ell} = \hist{\ell}} \ge \highprob.
\)
Hence, 
\begin{align*}
 \Pr{\bs_\ell \notin  \suit_\ell \mid \bhist{\ell} = \hist{\ell}} \le \frac{\epsP}{k^{10}}.
 \end{align*}

By above inequalities and the law of total expectation we have
\begin{align*}
      \Ex{{\bd_{\ell}} \mid \bhists{\ell} = \hists{\ell}}
      &=
      \Ex{{\bd_{\ell}} \mid \bhists{\ell} = \hists{\ell}, \ssize_\ell \in \suit_\ell}
      \Pr{\bs_\ell \in \suit_\ell \mid \bhist{\ell} = \hist{\ell}} \\
      &+
      \Ex{{\bd_{\ell}} \mid \bhists{\ell} = \hists{\ell}, \ssize_\ell \notin \suit_\ell}
      \Pr{\bs_\ell \notin \suit_\ell \mid \bhist{\ell} = \hist{\ell}} \\
      &\le
      \Ex{{\bd_{\ell}} \mid \bhists{\ell} = \hists{\ell}, \ssize_\ell \in \suit_\ell} \\
      &+
      \Ex{{\bd_{\ell}} \mid \bhists{\ell} = \hists{\ell}, \ssize_\ell \notin \suit_\ell}
      \Pr{\bs_\ell \notin \suit_\ell \mid \bhist{\ell} = \hist{\ell}} \\
      &\le 
      \frac{\epsDel(1+\epsBuck)}{1 - \tfrac{3}{2}\epsSamp}
      \Ex{\bw_\ell \mid \bhists{\ell} = \hists{\ell}} + 
      \frac{\epsP}{k^{10}}
     \Ex{{\bd_{\ell}} \mid \bhists{\ell} = \hists{\ell}, \ssize_\ell \notin \suit_\ell}.
\end{align*}

By definition, we know that $\bd_{\ell} \le w(\uni_T)$. Then, by Lemma~\ref{lm:wU<2wI}, it follows that $\bd_{\ell} \le 2\, w(\indep_T)$. Moreover, we know that $\indep_T$ contains at most $k$ elements and for each element $e$, we have $w(e) \le f(e) \le \ratio \, \tau$. Therefore, we conclude that $\bd_{\ell}$ is deterministically bounded by $2k \cdot \ratio \tau$.

% \[
% \Ex{{d_{\ell}} \mid \bhists{\ell} = \hists{\ell}, \ssize_\ell \notin \suit_\ell} \le \ratio \tau \ssize_\ell \le .
% \]
Meanwhile, we know that 
\(\Ex{\bw_{\ell} \mid \bhists{\ell} = \hists{\ell}} \ge \tau
\) since conditioning on \(\bhists{\ell} = \hists{\ell}\) implies that $\ell \leq \bT$, which ensures there is at least one element in \(\uni_\ell \setminus \uni_{\ell-1}\). Thus, we obtain the following bound: 
\[\Ex{{\bd_{\ell}} \mid \bhists{\ell} = \hists{\ell}, \ssize_\ell \notin \suit_\ell} \le \frac{8k^2}{\eps} \Ex{\bw_{\ell} \mid \bhists{\ell} = \hists{\ell}}. \]

This gives us 
\begin{align*}
      \Ex{{\bd_{\ell}} \mid \bhists{\ell} = \hists{\ell}}
      &\le 
      \frac{\epsDel(1+\epsBuck)}{1 - \tfrac{3}{2}\epsSamp}
      \Ex{\bw_\ell \mid \bhists{\ell} = \hists{\ell}} + 
      \frac{8\epsP}{k^{8}\eps}
     \Ex{\bw_{\ell} \mid \bhists{\ell} = \hists{\ell}}, 
\end{align*}
which establishes Equation \eqref{eq:modiftarg} and completes the proof.
% Therefore,
%   \begin{align*}
%       \Ex{{d_{\ell}} \mid \hist{\ell}}
%       &\le
%       \eps_1 \Ex{{w_{\ell}} \mid \bhists{\ell} = \hists{\ell}, \ssize_\ell \in \suit_\ell} + \epsilon \Ex{{w_{\ell}} \mid \bhists{\ell} = \hists{\ell}}
%       \\&\le 
%       (2\eps_1 + \eps) \Ex{{\bw_{\ell}} \mid \bhists{\ell} = \hists{\ell}},
%   \end{align*}
%   finishing the proof.
\end{proof}

Using the previous lemma, we extend the earlier guarantee and show that the expected weight of $\indep_T$ remains nearly unchanged after the removal of $\setd$.
\begin{lemma}\label{lm:approx:w_I_D__w_I}
The following holds: 
  \[(1-2\eps_2) \Ex{{w}({\indep_T})} \le  \Ex{{w}({\indep_T} \backslash D)}.\]
\end{lemma}

\begin{proof}
    Let $Z_T = \uni_T \setminus \indep_T$.
By Lemma \ref{lm:approx:w_U_D__w_U} we have:
$$\Ex{{w}({\uni_T} \backslash D)} \ge (1-\eps_2) \Ex{{w}({\uni_T})}.$$
Using the linearity of expectation for partition $\uni_T = \indep_T \cup Z_T$ gives us:
$$\Ex{{w}({\indep_T}\backslash D)} + \Ex{{w}({Z_T}\backslash D)} \ge (1-\eps_2) \left(\Ex{{w}({\indep_T})} + \Ex{{w}({Z_T})}\right).$$
By using the simple bound $\Ex{w(Z_T \backslash D)} \le \Ex{w(Z_T)}$, we obtain:
\begin{align}
    \Ex{{w}({\indep_T}\backslash D)} 
    &\ge (1-\eps_2) \Ex{{w}({\indep_T})} + (1-\eps_2) \Ex{{w}({Z_T})} - \Ex{{w}({Z_T})}
    \notag
    \\&\ge (1-\eps_2)\Ex{{w}({\indep_T})} - \eps_2 \Ex{{w}({Z_T})}.  
    \label{eq:oct19_2127}
\end{align}

By Lemma \ref{lm:wU<2wI} and then applying linearity of expectation to the partition $\uni_T = \indep_T \cup Z_T$ again, we have
  \begin{align*}
    2\Ex{{w}({\indep_T})} \ge \Ex{{w}({\uni_T})} = \Ex{{w}({\indep_T})} + \Ex{{w}({Z_T})}, 
  \end{align*}
which in turn implies that
$$\Ex{{w}({\indep_T})} \ge \Ex{{w}({Z_T})}.$$

Combining this and Equation \eqref{eq:oct19_2127} completes the proof:
\begin{align*}
    \Ex{{w}({\indep_T}\backslash D)} 
    & \ge (1-\eps_2)\Ex{{w}({\indep_T})} - \eps_2 \Ex{{w}({Z_T})}.  
    \\&\ge (1-\eps_2)\Ex{{w}({\indep_T})} - \eps_2 \Ex{{w}({\indep_T})}
    \\& \ge (1-2\eps_2)\Ex{{w}({\indep_T})}.
\end{align*}
\end{proof}

The next lemma bounds the weight of $\indep_T \setminus \setd$ by its function value.
\begin{lemma}\label{lm:approx:w_I_f_I}
The following holds:
  $w(\indep_T \setminus \setd) \le f(\indep_T \setminus \setd)$.
\end{lemma}
\begin{proof}
    Let \( e_1, e_2, \ldots, e_{\abs{\indep_T \setminus \setd}} \) denote the elements of \(\indep_T \setminus \setd\) in their processing order. We have 
    \[f(\indep_T \setminus \setd) = \sum_{i}{\relcont{e_i}{\{e_1, e_2, \ldots, e_{i - 1}\}}}. \] 

    Let's use \(\uni_{e_i}\) to denote the state of \(\uni\) right before processing element \(e_i\).
    By definition we have 
    \[w(\indep_T \setminus \setd) = \sum_{i}{w(e_i)}
    = \sum_{i}\relcont{e_i}{\uni_{e_i}}
    ,\]
    where the second equality holds because of Lemmas \ref{lm:w_withinlevel} and \ref{w=delt}. 

     We know that \(\{e_1, e_2, \ldots, e_{i - 1}\} \subseteq \uni_{e_i}\). Therefore, by submodularity of $f$ we have 
    \[\relcont{e_i}{\uni_{e_i}} \le \relcont{e_i}{\{e_1, e_2, \ldots, e_{i - 1}\}}, \] 
    which completes the proof.
\end{proof}

Finally, by combining the established results, we prove the approximation guarantee of our algorithm in the following theorem.

\begin{theorem}\label{thm:approx:main}
There exists an instance $\mathcal{I}$ produced by Algorithm~\ref{alg:parallel} such that
\[
\mathbb{E}[f(\indep_T \setminus D)] \ge \frac{1}{4}\left(1 - \frac{\eps}{2}\right)(1-2\eps_2) f(\optInd),
\]
where $\optInd$ is the optimal independent set and the expectation is taken over the randomness of the algorithm. Equivalently, the algorithm achieves an expected $(4 + O(\epsilon))$-approximation.
\end{theorem}

\begin{proof}
    Lemma \ref{lm:f_I*<4I} guarantees the existence of an instance $\mathcal{I}$ in Algorithm \ref{alg:parallel} such that
    \[
    \frac{1}{4}\left(1 - \frac{\eps}{2}\right) f(\optInd) < w(\indep_T).
    \]
    By Lemma \ref{lm:approx:w_I_D__w_I}, we further have
    \[
    (1-2\eps_2) \Ex{w(\indep_T)} \le \Ex{w(\indep_T \setminus D)},
    \]
    and by Lemma \ref{lm:approx:w_I_f_I},
    \[
    w(\indep_T \setminus D) \le f(\indep_T \setminus D).
    \]
    Combining these inequalities yields an instance for which the following final bound holds:
    \[
    \frac{1}{4}\left(1 - \frac{\eps}{2}\right) (1-2\eps_2) f(\optInd) \le \Ex{f(\indep_T \setminus D)}.
    \]
    The proof is complete.
\end{proof}

\subsection{Query Time}
\label{sec:query:time}
\newcommand{\relem}{\overline{\elem}}
In this section, we analyze the query complexity of our algorithm. 

We begin by introducing a potential-based framework to track the progress of elements across layers, which enables us to bound the number of layers. 
%The following definition formalizes the notion of potential, which serves as a measure of the remaining potential for contribution of elements at each layer.
\begin{definition}[Potential]
    For any $i \ge 1$, let \(\relem_{i}\) denote \((\elem_{i} \cup \add_{i}) \). 
   For any element \( e \in \relem_{i}\), we define the potential of element
   \( e \) with respect to the layer $i$ as:
    \[P(e, i) := \floor{\log_{1+\eps}(\frac{\relcont{e}{\uni_{i - 1}}}{\tau})} + 1. \]
  For an element \(e\notin \relem_{i}\), we define
  $P(e, i)$ to be zero.

    We also, define the potential of layer $i$ as:
  \[P_i := \sum_{e\in \relem_{i}} P(e, i).\]
  %We also define $P_i$ for $i> T+1$ to be 0.
\end{definition}

We then establish some fundamental properties of the potential function that will be used in our analysis.
\begin{lemma}\label{lm:potential_properties}
The followings hold:
  \begin{align*}
      &\forall e \in \relem_{i}: P(e, i)\in \left[1, \ceil{\log_{1+\epsBuck}(\ratio)}\right] 
      \\&
      |\relem_i| \le P_i \le |\relem_i| \cdot \ceil{\log_{1+\epsBuck}(\ratio)}
      \\&
      P_{T + 1} = 0 % \quad\text{ and }\quad \forall i \le T: P_{i} > 0
  \end{align*}
\end{lemma}

\begin{proof}

    For any element $e \in \relem_{1}$, we know that $\tau \le f(e) < \tau \cdot \ratio$. 
    Additionally, for any $i > 1$, for any element $e$ in either  \( \elem_{i}\) or \(\add_i\) , we know that \(\inspect(e, \trip_{i - 1}, \tau) \ne \err\).  Hence, 
    \(\relcont{e}{\uni_{i - 1}} > \tau. \) 
    Therefore, the first property always holds.
    
    The second property follows from the first one since
    $P(e, i) = 0$ for $e\notin \celem_{1}$.

    We know by Lemma \ref{lm:inv:finish} that \(\elem_{T + 1} = \emptyset\). Moreover, Lemma \ref{lm:reconstruction_condition} implies that \(\add_{T + 1} = \emptyset\). Hence, \(\relem_{T + 1} = \emptyset\). Therefore, the third property also holds. 
\end{proof}

Next, we show that the potential of an element is decreasing through the layers.

\begin{lemma}\label{lm:potential_decrease_always}
    For any $i \in [1, T]$, and any arbitrary element $e$,
    $P(e, i) \ge P(e, i+1)$ and therefore,
    $P_{i} \ge P_{i+1}$.
\end{lemma}
\begin{proof}
  If $e \notin \relem_{i + 1}$, the claim holds because $P(e, i + 1) = 0$. 
  It also holds for $e \in \relem_{i + 1}$ due to the submodularity of $f$ and the fact that $\uni_i \subseteq \uni_{i + 1}$. 
\end{proof}

The following is an auxiliary fact that will be used to bound the expected number of layers in our data structure.
\begin{fact}\label{lm:auxil}
    Let $\bX_0, \bX_1, \dots, \bX_n$ be a sequence of integer positive variables such that $\bX_{i} \le \bX_{i-1}$ and 
    \begin{align*}
        \Ex{\bX_i \mid \bX_1 = X_1, \dots, \bX_{i-1} = X_{i-1}}
        \le 
        (1-\eps') {X_{i-1}}. 
    \end{align*}
    Let $T$ denote the first index $i$ such that $X_T = 0$ and assume that $\bX_0=N$ for some fixed integer $N$. Then
    $\Ex{\bT} \le \frac{\log(N) + 1}{\poly(\eps')}$.
\end{fact}

The following two lemmas quantify the expected decrease in potential across consecutive layers, showing that, in expectation, a substantial fraction of the potential diminishes at each layer. 
\begin{lemma}
  For any value of $i \ge 1$, the following holds:
  \begin{align*}
  \Ex{P_{i} - \bP_{i+1} \big| \bhists{i}=\hists{i}} \ge 
  \frac{\epsSamp}{2} \cdot \abs{\buck_{i}} \cdot
    \ind{\ssize_i \in \suit_i}
    .
  \end{align*}
\end{lemma}
\begin{proof}
    We know that $P(e, i) = 0$ for any $e \notin \relem_i$. Hence, 
  \begin{align*}
  P_{i} - P_{i+1} 
  &= \sum_{e \in \relem_{i}} P(e, i) - P(e, i + 1) \\
    & \ge \sum_{e \in \relem_{i}} \ind{P(e, i) - P(e, i+1) \ge 1}\\
    &= \sum_{e \in \buck_{i}} \ind{P(e, i) - P(e, i+1) \ge 1}
    + 
    \sum_{e \in \relem_{i} \setminus \buck_{i}}
    \ind{P(e, i) - P(e, i+1) \ge 1} \\
    &\ge \sum_{e \in \buck_{i}} \ind{P(e, i) - P(e, i+1) \ge 1}.  %\\
    %& = \sum_{e \in \buck_i \setminus \setd} \ind{P(e, i) - P(e, i+1) \ge 1}.
  \end{align*}
  The inequality comes from the fact that $P(e, \ell)$ is decreasing by Lemma \ref{lm:potential_decrease_always} and an integer by definition, so \(P(e, i) - P(e, i + 1)\) is either 0 or it can be bounded by 1.
  %The last equality comes from the fact that 
  %we know that \(\buck_i \subseteq \elem_i\), so \((\add_i \cup \elem_i) \cap B_i = \buck_i\). 

    We know that if \(e \in \buck_i\), then \(P(e, i) \ge 1\).
    We also know that if \(\inspect(e, \trip_i, \tau) = \err\), then \(e \notin \relem_{i + 1}\), 
    which means that \(P(e, i + 1) = 0\). Hence, for any \(e \in \buck_i\) such that \(\inspect(e, \trip_i, \tau) = \err\), \(\ind{P(e, i) - P(e, i+1) \ge 1} = 1\). 

    By Lemma \ref{lm:suit1impliessuit2}, we have assuming that \(\ssize_i \in \suit_i\), then 
    \[
    \Ex{
        \abs{\cbr{e \in \buck_i: \inspect{}(e, \trip_{i}, \tau) = \err}}} \ge (\frac{\epsSamp}{2})\abs{\buck_i}.\]
    The proof is complete.

\end{proof}

\begin{lemma}\label{lm:exp:P:margin}
 For any value of $i \ge 1$, the following holds:
  \begin{align*}
  \Ex{P_{i} - \bP_{i+1} \big| \bhist{i}=\hist{i}} \ge 
  (\highprob) \cdot \frac{\epsSamp}{2} \cdot \abs{\buck_{i}} \ge 
  \frac{\epsSamp}{4} \cdot \abs{\buck_{i}}.
  \end{align*}
\end{lemma}
\begin{proof}
  By the law of iterated expectation 
  and the previous Lemma, we have
  \begin{align*}
  \Ex{P_{i} - \bP_{i+1} \big| \bhist{i}=\hist{i}} 
  &=
    \Ex{\Ex{P_{i} - \bP_{i+1} \big| \bhists{i}=\hists{i}}} \\
    %\Ex{\Ex{P_{i} - \bP_{i+1} \big| \bhist{i}=\hist{i}, \bs_\ell }}
    &\ge \Ex{\frac{\epsSamp}{2} \cdot \abs{\buck_{i}} \cdot
    \ind{\ssize_i \in \suit_i}}\\
    &\ge
     (\highprob) \cdot \frac{\epsSamp}{2} \cdot \abs{\buck_{i}}
  \end{align*}
  The proof is complete. 
\end{proof}

Using the above results, this lemma provides a polylogarithmic bound on the number of remaining layers. 
\begin{lemma}\label{cor:bound_T_invariant}
For any $\ell \ge 1$,
  \begin{align*}
    \Ex{\bT - i \mid \bhist{i}=\hist{i}} \le \poly(\log(\abs{\relem_i}), \log(k) , \frac{1}{\eps}).
  \end{align*}
\end{lemma}
\begin{proof}
    The elements of \(\elem_i\) are divided into \(
    \floor{\log_{(1+\epsBuck)}(\ratio)}\) buckets, and \(\buck_i\) is the largest one among them. We also know from Lemma \ref{lm:reconstruction_condition} that \(|\relem_i| \leq (1 + \epsIns) |\elem_i|\). 
    Therefore, 
    \[\abs{\buck_i} \ge \frac{|\elem_i|}{\floor{\log_{(1+\epsBuck)}(\ratio)}} \cdot \frac{|\relem_i|}{|\elem_i|} = \frac{|\relem_i|}{\floor{\log_{(1+\epsBuck)}(\ratio)}} \] 

    By Lemma \ref{lm:exp:P:margin}, we have 
    \[\Ex{P_{i} - \bP_{i+1} \big| \bhist{i}=\hist{i}} \ge (\highprob) \cdot \frac{\epsSamp}{2} \cdot \frac{|\relem_i|}{\floor{\log_{(1+\epsBuck)}(\ratio)}}\]
    This, the second property of Lemma \ref{lm:potential_properties}, and Fact \ref{lm:auxil}, completes the proof. 
    
\end{proof}

We are now ready to bound the expected total number of layers.

\begin{corollary}
\label{lm:num_level_total}
  At any point in the stream,
  the expected number of layers $\Ex{\bT}$ is at most $\poly(\log(n), \log(k), \frac{1}{\eps})$.
\end{corollary}
\begin{proof}
By Lemma \ref{cor:bound_T_invariant}, we have
      \[
    \Ex{\bT - 1} = \Ex{T - 1 \mid \bhist{1}=\hist{1}} \le \poly(\log(n), \frac{1}{\eps}).
  \]
Howevere note that, the expectation in the result above is over all of the randomness of the algorithm. Yet, since for each update \(u\) the value of \(T)\) after the execution update \(u\) is independent of the randomness of the algorithm during subsequent updates, the same holds for expectation over the randomness of the algorithm until the execution of \(u\) finishes. 
\end{proof}

Given the upper bound on the number of layers, we now compute the query complexity of procedure $\construct$.
\begin{lemma}\label{lm:bound_num_query_recounstruct}
Suppose the procedure \(\construct(i)\) is invoked. The expected number of queries made by calling \(\construct{}(i)\) is
  $|\relem_i| \cdot \poly(\log(n), \log(k), \frac{1}{\eps}) $, where $|\relem_i|$ refers to the size of $\relem_i$ after the update, which is equal to $\celem_i$ right before executing \(\construct(i)\), and the expectation is over the randomness of algorithm only during the execution of \(\construct(i)\). 
  % \pjcomment{we can ommit $\log(|\Rc_i|)$}
\end{lemma}
\begin{proof}
    We know that conditioned on history up to level \(i\), the expected number of leyers after \(i\) is polylogarithmic. The number of queries made during \(\estimatesize(j \ge i)\) is bounded by 
    \(\abs{\relem_i} \cdot \poly((\log(n), \log(k), \frac{1}{\eps}))\). Hence, the number of total queries is bounded by $|\relem_i| \cdot \poly(\log(n), \log(k), \frac{1}{\eps}) $. 
\end{proof}

%Finally, we compute the expected number of queries for each update.

Finally, by combining the above results and introducing a charging scheme for the cost, we distribute the cost of each invocation of $\construct$ across updates, thereby bounding the total query complexity of our algorithm per each update.
\begin{theorem}
\label{thm:update:time:amorti}
The expected amortized update time of our entire algorithm is 
\(\poly(\log(n), \log(k), \eps^{-1})\).
\end{theorem}

\begin{proof}

Each update affects only \(\log(\ratio)\) instances of the main Algorithm. Hence, it suffices to bound the number of queries per update for a fixed instance of the Main Algorithm \ref{alg:main}. 

The number of query calls made directly via \(\ins(e)\) or \(\del(e)\) is bounded by \(T\), which by Corollary \ref{lm:num_level_total} is known to be polylogarithmic. 

To upper bound the number of indirect query calls made per update through the procedure \(\construct\), we employ the following charging scheme:

Let's refer to the number of calls made during the execution of the procedure $\construct(\ell)$ as its cost. Note that this cost also reflects the indirect queries that $\construct(\ell)$ makes through $\construct(\ell + 1)$. 
We charge the cost of executing $\construct{}(\ell)$ to all the updates that were responsible for triggering its invocation. 
In particular, if the procedure \(\ins(e)\) directly invokes $\construct{}(\ell)$ (not through \(\construct{}(\ell - 1)\)), we spread the cost of $\construct{}(\ell)$ among the updates corresponding to the elements in \(\add_\ell\) immediately before its execution. Likewise, if \(\del(e)\) directly invokes $\construct{}(\ell)$, we charge its cost to all updates corresponding to the elements in \(\buck_{\ell} \cap \setd\) immediately before its execution. 

For updates $u$ and $v$ and layer $i$, 
let $Q_{u, i, v}$ denote the cost charged to $u$ due to the invocation of $\construct(i)$ during the execution of update $v$. Note that if $\construct(i)$ was not invoked or if it was invoked indirectly through \(\construct(i - 1)\), then $Q_{u, i, v} = 0$.

Consider a fixed \(u\) and let $T^u$ denote the number of levels right before update $u$. We have:

\[\Ex{\bQ_{u, i, v}}
      = \Ex{\bQ_{u, i, v} \mid \bQ_{u, i, v} > 0}\Pr{\bQ_{u, i, v} > 0}\]

Given \(\bQ_{u, i, v} > 0\), we know that \(\construct(i)\) has been invoked during update $v$. By Lemma \ref{lm:bound_num_query_recounstruct} we know that, the value of \(\bQ_{u, i, v}\) is independent of the randomness of the algorithm before the execution of  \(\construct(i)\) during update $v$ and it is bounded by \(\poly(\log(n), \log(k), \frac{1}{\epsilon}) \)  as the cost of \( \construct(i) \) has been spread bacross \( O(\abs{\relem_i}.\eps^{-1}) \) different updates. We also note that whether \(\bQ_{u, i, v} = 0\) or not can be deterministically determined by the randomness of the algorithm 
up to starting the execution of update $v$. Therefore, to prove that 
\(\sum_{i, v}\Ex{\bQ_{u, i, v} }\) is polylogarithmic, it only remains to show that 
\begin{align}
\label{eqhre}
      \sum_{i, v} \Pr{\bQ_{u, i, v} > 0 } \le \poly(\log(n), \log(k), \frac{1}{\epsilon})
\end{align}

We know that the element corresponding to update $u$ may only be involved in the invocation of \(\construct(i)\) during update $v$ if update $v$ is the first update since receiving update $u$ that has caused a direct or indirect invocation of  \(\construct(i)\), since  \(\add_i\) and \(\buck_{i} \cap \setd\) would become empty each time \(\construct(i)\) is executed. The same reasoning also implies that \(\bQ_{u, i, v}\) cannot be nonzero for any $i > \bT^u$ as the existence of such level can only be due to the execution of \(\construct(i)\) during or after update $u$. Therefore, we have: 
\[
\Ex{\sum_{i, v} \ind{\bQ_{u, i, v} > 0} } \le \Ex{\bT^{u} + 1}.
\]

Since \(\Ex{\sum_{i, v} \ind{\bQ_{u, i, v} > 0} } = \sum_{i, v} \Pr{\bQ_{u, i, v} > 0 }\), Corollary \ref{lm:num_level_total} proves inequality \eqref{eqhre} and completes the proof.

\end{proof}

\bibliographystyle{plain}
\bibliography{references}

\newpage
\appendix
\section{Omitted proofs}

\begin{lemma} %todo change
\label{lm:unigrws}
Given $0 \leq i < j \leq T$, we have $\uni_i \subseteq \uni_j$.
\end{lemma}

\begin{proof}
We know from Lemma~\ref{lm:invar:expand} that for all $\ell \in [T]$,
\[\trip_{\ell} = \expand(\samp_{\ell}, \trip_{\ell-1}, \tau_{\ell}).\]
In the $\expand$ procedure (Algorithm~\ref{alg:expand}), the set $\uni_\ell$ is initially set to $\uni_{\ell - 1}$ and may be augmented during the sequence of $\adjust$ operations on elements in $\samp_{\ell}$. However, no elements are removed from it during this process.

Therefore, we have $\uni_{\ell - 1} \subseteq \uni_\ell$ for all $\ell \in [T]$, and by applying this inductively it follows that $\uni_i \subseteq \uni_j$ for all $0 \leq i < j \leq T$, as desired.
\end{proof}

\begin{lemma}%[\cite{DBLP:books/daglib/0070636}]
\label{lem:elimination:axiom}
The family $\mathcal{C}$ of circuits of a matroid $\matroid$ has the following properties:
\begin{itemize}
    \item ($C1$) $\emptyset \notin \mathcal{C}$.
    \item ($C2$) if $C_1$, $C_2 \in \mathcal{C}$ and $C_1 \subseteq C_2$, then $C_1=C_2$.
    \item ($C3$)  if $C_1$, $C_2 \in \mathcal{C}$, $C_1 \neq C_2$ and $e\in C_1 \cap C_2$, then there exists $C_3 \in \mathcal{C}$ such that $C_3 \subseteq C_1 \cup C_2 \setminus \{e\}$.
\end{itemize}

\end{lemma}

\begin{lemma}
\label{lem:=1circuit}
Let \( e \) and \( I \) be an arbitrary element and independent set in \( \mathcal{M} \). If \( I \cup \{e\} \notin \mathcal{I} \), then \( I \cup \{e\} \) contains a unique circuit \(C\).
\end{lemma}

\begin{proof}
Starting with the dependent set \( I \cup \{e\} \), one can iteratively remove an element that its removal preserves dependence. 
Since \( \emptyset \in \mathcal{I} \), this process must terminate with a minimal dependent set, which is, by definition, a circuit. Therefore, there exists at least one circuit \( C \subseteq I \cup \{e\} \).

Suppose, for contradiction, that there are two distinct circuits \( C_1, C_2 \subseteq I \cup \{e\} \) with \( C_1 \neq C_2 \). 
Since \( I \) is independent, neither \( C_1 \) nor \( C_2 \) can be fully contained in \( I \). Thus, \( e \in C_1 \cap C_2 \).
By Lemma~\ref{lem:elimination:axiom} there exists a circuit \( C_3 \subseteq (C_1 \cup C_2) \setminus \{e\} \). 
However, \( (C_1 \cup C_2) \setminus \{e\} \subseteq I \), implying \( C_3 \subseteq I \), which is in contradiction with the independence of \( I \).

Therefore, \( I \cup \{e\} \) contains exactly one circuit.
\end{proof}

\newcommand{\rank}{rank}

\begin{lemma}\label{lm:find_min} Let $I \in \mI$ be an independent set and $e$ be an element such that $I \cup \{e\} \notin \mI$.
  Define $C := \{e': I + e - e' \in \mI\}$.
  Let $w : I \cup \{e\} \to \mathbb{R}^{\ge 0}$ be an arbitrary weight function and
  define $\hat{e} := \argmin_{e'\in C} w(e')$. The element $\hat{e}$ can be found
  using at most $O(\log(|I|))$ oracle queries.
\end{lemma}
\begin{proof}
  Let $e_1, \dots, e_{|I| + 1}$ denote an ordering of $I \cup \{e\}$
  such that $w(e_1) \ge w(e_2) \dots \ge w(e_{|I| + 1})$. 
  Let $i$ denote the smallest index such that
  $\{e_1, \dots, e_{i}\} \notin \mI$. Such an index exists
  because $\{e_1, \dots, e_{|I| + 1}\} = I \cup \{e\} \notin \mI$. 
  We claim that $\hat{e} = e_i$.
  We note that the element $e_i$ can be found
  using a binary search over $[|I| + 1]$ because
  for any $j$,
  if $\{e_1, \dots, e_j\} \notin \mI$, then
  $\{e_1, \dots, e_{j+1}\} \notin \mI$ as well.

  To prove this, we first claim that $e_i \in C$. To
  see why this holds, we first observe that
  since $\{e_1, \dots, e_{i-1}\}$ is independent but
  $\{e_1, \dots, e_i\}$ is not, we have
  $e_i \in \spn(\{e_1, \dots, e_{i-1}\}) \subseteq \spn(I + e - e_i)$.
  Therefore, since $e_j \in \spn(I + e - e_i)$ for all $j\ne i$, we have
  $I + e \subseteq \spn(I + e - e_i)$, which implies
  \begin{align*}
    \rank(I + e - e_i) \ge \rank(I + e) \ge \rank(I) = |I|
    = |I+ e - e_i|,
  \end{align*}
  which implies $I + e - e_i \in \mI$ as claimed.

  We need to show that for any $e' \in C$, we have $w(e_i) \le w(e')$. 
  Assume for contradiction that $w(e') < w(e_i)$. It follows that
  $e' = e_j$ for some $j > i$. By definition of $C$,
  we must have
  $I + e - e_j \in \mI$, which implies
  $\{e_1, \dots, e_{j-1}\} \in\mI$. Since $i < j$, this further implies
  $\{e_1, \dots, e_i\} \in  \mI$, which is not possible by definition of $i$.
\end{proof}

\end{document}